\documentclass[12pt]{amsart}

\usepackage{amsthm}
\usepackage{amssymb}
\usepackage{amsmath}

\numberwithin{equation}{section}

\newtheorem{thm}{Theorem}[section]
\newtheorem{prop}[thm]{Proposition}
\newtheorem{lem}[thm]{Lemma}
\newtheorem{cor}[thm]{Corollary}

\theoremstyle{definition}

\newtheorem{defn}[thm]{Definition}
\newtheorem{rem}[thm]{Remark}
\newtheorem{ex}[thm]{Example}

\begin{document}

\title[On the eigenfunctions for the multi-species $q$-Boson system]
{On the eigenfunctions for the multi-species $q$-Boson system}
\author{Yoshihiro Takeyama}
\address{Division of Mathematics, 
Faculty of Pure and Applied Sciences, 
University of Tsukuba, Tsukuba, Ibaraki 305-8571, Japan}
\email{takeyama@math.tsukuba.ac.jp}

\begin{abstract}
In a previous paper a multi-species version of the $q$-Boson stochastic particle system is introduced and 
the eigenfunctions of its backward generator are constructed by using a representation of 
the Hecke algebra. 
In this article  we prove a formula which expresses the eigenfunctions 
by means of the $q$-deformed bosonic operators, 
which are constructed from 
the $L$-operator of higher rank found in the recent work by Garbali, de Gier and Wheeler.  
The $L$-operator is obtained from 
the universal $R$-matrix of the quantum affine algebra of type $A_{r}^{(1)}$ 
by the use of the $q$-oscillator representation. 
Thus our formula may be regarded as a bridge between two approaches to studying 
integrable stochastic systems by means of the quantum affine algebra and the affine Hecke algebra. 
\end{abstract}

\maketitle

{\small
\noindent{\it 2010 Math. Subj. Class.} 
17B80, 81R12, 60K35.}
\medskip 

\setcounter{section}{0}
\setcounter{equation}{0}


\section{Introduction}

In this article we prove a formula which expresses 
the eigenfunctions of the backward generator of 
a multi-species version of the $q$-Boson system  
(which is constructed in the previous paper \cite{T}) 
by means of $q$-deformed bosonic operators. 

The $q$-Boson system is a stochastic particle system 
on the one-dimensional discrete space with continuous time \cite{SW}. 
It is a zero-range process for bosonic particles, 
where one particle may hop to the left nearest-neighboring site independently for each site. 
It is integrable in the following sense. 
There exists the $L$-operator whose entries are $q$-deformed bosonic operators and 
the transfer matrix, which is the trace of a product of the $L$-operators, 
generates mutually commuting operators, one of which is the transition rate matrix 
of the $q$-Boson system with periodic boundary condition. 

A multi-species version of the $q$-Boson system is introduced in \cite{T}.   
It is also a zero-range process for bosonic particles, 
though each particle is colored with a positive integer and  
the transition rate depends on the color. 
As we will see in Section \ref{subsec:tranfer-matrix} below,  
using the $L$-operator of higher rank found in the recent work of Garbali, de Gier and Wheeler \cite{GGW},  
we can construct a family of commuting operators 
which includes the transition rate matrix of the multi-species $q$-Boson system with periodic boundary. 
In this sense the system is also integrable. 
Recently it has been clarified that it is related to a representation theory 
of the quantum affine algebra \cite{Kuan, KMMO}. 

Originally, the multi-species $q$-Boson system was constructed from 
a representation of a deformation of the affine Hecke algebra of type $GL$. 
In this algebraic setting we can easily construct as a byproduct 
the eigenfunctions of the backward generator of the system. 
They are expressed by means of an action of the Hecke algebra of type $A_{k-1}$ 
on the vector space $(\mathbb{C}^{r})^{\otimes k}$, where 
$k$ is the total number of particles and $r$ is the number of species of particles. 
When $r=1$ they are equal to the eigenfunctions found by 
Borodin, Corwin, Petrov and Sasamoto \cite{BCPS}. 
A similar formula for the configuration probability in the multi-species asymmetric exclusion process 
is obtained by Tracy and Widom \cite{TW}. 

The purpose of this paper is to prove a formula which expresses the eigenfunctions 
for the multi-species $q$-Boson system by means of $q$-deformed bosonic operators. 
The formula is of the following form 
\begin{align*}
\lim_{\begin{subarray}{c} M \to \infty \\ M' \to -\infty \end{subarray}}
\prod_{i=1}^{k}\frac{z_{i}^{M'-1}}{(1+z_{i})^{M}} 
\langle 
\prod_{1\le i \le k}^{\curvearrowright}C_{\mu_{i}}^{[M', M]}(z_{i})
\prod_{1\le i \le k}\beta_{\nu_{i}, x_{i}}^{*}
\rangle_{[M', M]}.     
\end{align*} 
It is a function on the set of configurations of $k$ particles with $r$ species on $\mathbb{Z}$, 
and $(x_{1}, \ldots , x_{k})$ and $(\nu_{1}, \ldots , \nu_{k})$ signify 
the positions and the colors of the particles, respectively. 
The operator $C_{a}^{[M', M]}(z) \, (1\le a \le r)$ is a $q$-deformed bosonic operator produced from 
an ordered product of the $L$-operators assigned to 
the sites on the interval $[M', M]\subset \mathbb{Z}$, 
and $\beta_{a, i}^{*}$ is the creation operator assigned to the $i$-th site. 
The eigenfunctions are parametrized by two tuples $(z_{1}, \ldots , z_{k}) \in \mathbb{C}^{k}$ and 
$(\mu_{1}, \ldots , \mu_{k}) \in \{1, 2, \ldots , r\}^{k}$ such that 
the number of $\mu_{i}$'s satisfying $\mu_{i}=a$ is equal to that of particles with color $a$ 
($1\le a \le r$). 
The symbol $\langle \quad \rangle_{[M', M]}$ represents the matrix element 
between the vacuum vector and its dual on the interval $[M', M]$. 
 
As shown in \cite{GGW}, the $L$-operator of higher rank, which is a source of 
the $q$-deformed bosonic operator $C_{a}^{[M', M]}(z)$, can be constructed from the universal $R$-matrix 
of the quantum affine algebra of type $A_{r}^{(1)}$ 
by the use of the $q$-oscillator representation. 
On the other hand, the eigenfunctions are constructed by using 
a representation of the deformed affine Hecke algebra in \cite{T}. 
Thus our formula above may be regarded as a bridge between two approaches to studying 
integrable stochastic systems by means of the quantum affine algebra and the affine Hecke algebra. 

The paper is organized as follows. 
In Section \ref{sec:q-Boson} we give the definition of the multi-species $q$-Boson system. 
In Section \ref{sec:eigen} the algebraic construction of the eigenfunctions is explained following \cite{T}. 
In Section \ref{sec:q-boson-alg} we introduce the algebra of $q$-deformed bosons and 
its Fock representation, and discuss the integrability of 
the multi-species $q$-Boson system with periodic boundary. 
In Section \ref{sec:main} we state and prove the main theorem.  

Throughout this paper we denote by 
$\prod_{i \in I}^{\curvearrowright}A_{i}$ or $\prod_{a \le i \le b}^{\curvearrowright}A_{i}$
the ordered product $A_{a}A_{a+1}\cdots A_{b-1}A_{b}$ of non-commutative elements $A_{i}$ labeled with 
integers in the interval $I=[a, b]$. 
Similarly we write $\prod_{i \in I}^{\curvearrowleft}A_{i}$ or $\prod_{a \le i \le b}^{\curvearrowleft}A_{i}$  
for the reverse ordered product $A_{b}A_{b-1}\cdots A_{a+1}A_{a}$. 
We often abbreviate the sequence of positive integers $(a, \ldots , a)$ of length $m$ to $a^m$. 
For example, we write $(3, 3, 2, 2, 2, 1)$ as $(3^{2}, 2^{3}, 1^{1})$.


\section{The multi-species $q$-Boson system}\label{sec:q-Boson}

Hereafter we fix a real number $0<q<1$ and a positive integer $r$. 
The multi-species $q$-Boson system is a stochastic particle system 
on one-dimensional discrete space with continuous time. 
We consider the system on the infinite lattice $\mathbb{Z}$. 
The number of particles is finite and 
each particle is colored with a positive integer which is less than or equal to $r$.  
The particles can occupy the same site simultaneously. 
One particle may move to the left nearest-neighboring site independently for each site. 
The transition rate is given by
\begin{align}
\frac{1-q^{2n_{a}}}{1-q^{2}}q^{2\sum_{p=a+1}^{r}n_{p}},  
\label{eq:rate}
\end{align}
where $a$ is the color of the moving particle and 
$n_{p} \, (1\le p \le r)$ is the number of particles with color $p$ in the cluster from which 
the moving particle leaves. 
For example the rate of the transition described in Figure \ref{fig:fig1} is equal to 
$(1-q^{2\cdot 2})q^{2(1+1)}/(1-q^2)$ because $a=2, n_{2}=2$ and $n_{3}=n_{4}=1$.  

\begin{figure}[htbp]
{\unitlength 0.1in%
\begin{picture}(30.0000,15.0000)(4.0000,-18.0000)%
%
\special{pn 8}%
\special{pa 400 1800}%
\special{pa 1600 1800}%
\special{fp}%
%
\special{pn 8}%
\special{ar 800 1600 100 100 0.0000000 6.2831853}%
\put(8.0000,-16.0000){\makebox(0,0){$\scriptsize{1}$}}%
%
\special{pn 8}%
\special{ar 1190 1600 100 100 0.0000000 6.2831853}%
\put(11.9000,-16.0000){\makebox(0,0){$\scriptsize{4}$}}%
%
\special{pn 8}%
\special{ar 1190 1300 100 100 0.0000000 6.2831853}%
\put(11.9000,-13.0000){\makebox(0,0){$\scriptsize{3}$}}%
%
\special{pn 8}%
\special{ar 1190 1000 100 100 0.0000000 6.2831853}%
\put(11.9000,-10.0000){\makebox(0,0){$\scriptsize{2}$}}%
%
\special{pn 8}%
\special{ar 1190 700 100 100 0.0000000 6.2831853}%
\put(11.9000,-7.0000){\makebox(0,0){$\scriptsize{2}$}}%
%
\special{pn 8}%
\special{ar 1190 400 100 100 0.0000000 6.2831853}%
\put(11.9000,-4.0000){\makebox(0,0){$\scriptsize{1}$}}%
%
\special{pn 8}%
\special{pa 2200 1800}%
\special{pa 3400 1800}%
\special{fp}%
%
\special{pn 8}%
\special{ar 2600 1600 100 100 0.0000000 6.2831853}%
\put(26.0000,-16.0000){\makebox(0,0){$\scriptsize{2}$}}%
%
\special{pn 8}%
\special{ar 2990 1600 100 100 0.0000000 6.2831853}%
\put(29.9000,-16.0000){\makebox(0,0){$\scriptsize{4}$}}%
%
\special{pn 8}%
\special{ar 2990 1300 100 100 0.0000000 6.2831853}%
\put(29.9000,-13.0000){\makebox(0,0){$\scriptsize{3}$}}%
%
\special{pn 8}%
\special{ar 2990 1000 100 100 0.0000000 6.2831853}%
\put(29.9000,-10.0000){\makebox(0,0){$\scriptsize{2}$}}%
%
\special{pn 8}%
\special{ar 2590 1300 100 100 0.0000000 6.2831853}%
\put(25.9000,-13.0000){\makebox(0,0){$\scriptsize{1}$}}%
%
\special{pn 8}%
\special{ar 2990 700 100 100 0.0000000 6.2831853}%
\put(29.9000,-7.0000){\makebox(0,0){$\scriptsize{1}$}}%
%
\special{pn 13}%
\special{pa 1800 1000}%
\special{pa 2000 1000}%
\special{fp}%
\special{sh 1}%
\special{pa 2000 1000}%
\special{pa 1933 980}%
\special{pa 1947 1000}%
\special{pa 1933 1020}%
\special{pa 2000 1000}%
\special{fp}%
\end{picture}}%
\caption{}
\label{fig:fig1}
\end{figure}

We denote by $k_{a} \, (1\le a \le r)$ the number of particles with color $a$, 
and by $k=\sum_{a=1}^{r}k_{a}$ the total number of particles. 
They are conserved in the process. 
We introduce two sets 
\begin{align*}
I_{k_{1}, \ldots , k_{r}}=\left\{ 
\vec{\nu}=(\nu_{1}, \ldots , \nu_{k}) \in \{1, \ldots , r\}^{k} \, | \, 
\#\{j \, | \nu_{j}=a\}=k_{a} \, (1\le \forall{a} \le r)
\right\}  
\end{align*}
and 
\begin{align*}
L_{k}^{+}=\left\{ 
\vec{x}=(x_{1}, \ldots , x_{k}) \in \mathbb{Z}^{k} \, | \, x_{1} \ge \cdots \ge x_{k}
\right\},    
\end{align*}
and define 
\begin{align*}
\mathcal{S}_{k_{1}, \ldots , k_{r}}=\{ (\vec{x}, \vec{\nu}) \in L_{k}^{+}\times I_{k_{1}, \ldots , k_{r}} \, | \, 
\hbox{For any $1\le i \le k$, if $x_{i}=x_{i+1}$ then $\nu_{i}\le \nu_{i+1}$.}
\}.  
\end{align*}
We identify $\mathcal{S}_{k_{1}, \ldots , k_{r}}$ 
with the set of configurations of $k_{a}$ bosonic particles with color $a \, (1\le a \le r)$ on $\mathbb{Z}$ 
by assigning to $(\vec{x}, \vec{\nu})$ the configuration 
such that the particles with color $\nu_{1}, \ldots , \nu_{k}$ are on the sites $x_{1}, \ldots , x_{k}$, 
respectively. 
For example, the configuration in Figure \ref{fig:fig2} corresponds to the pair of 
$\vec{x}=(2, 2, 2, 0, -1, -1)$ and $\vec{\nu}=(1, 1, 3, 5, 2, 4)$. 

\begin{figure}[htbp]
{\unitlength 0.1in%
\begin{picture}(20.1000,10.3500)(3.9000,-11.3500)%
%
\special{pn 8}%
\special{ar 790 800 100 100 0.0000000 6.2831853}%
\put(7.9000,-8.0000){\makebox(0,0){$\scriptsize{4}$}}%
%
\special{pn 8}%
\special{ar 790 500 100 100 0.0000000 6.2831853}%
\put(7.9000,-5.0000){\makebox(0,0){$\scriptsize{2}$}}%
%
\special{pn 8}%
\special{ar 1190 800 100 100 0.0000000 6.2831853}%
\put(11.9000,-8.0000){\makebox(0,0){$\scriptsize{5}$}}%
%
\special{pn 8}%
\special{ar 1990 500 100 100 0.0000000 6.2831853}%
\put(19.9000,-5.0000){\makebox(0,0){$\scriptsize{1}$}}%
%
\special{pn 8}%
\special{pa 400 1000}%
\special{pa 2400 1000}%
\special{fp}%
\put(8.0000,-12.0000){\makebox(0,0){$\scriptsize{-1}$}}%
\put(12.0000,-12.0000){\makebox(0,0){$\scriptsize{0}$}}%
\put(16.0000,-12.0000){\makebox(0,0){$\scriptsize{1}$}}%
\put(20.0000,-12.0000){\makebox(0,0){$\scriptsize{2}$}}%
%
\special{pn 8}%
\special{ar 1990 800 100 100 0.0000000 6.2831853}%
\put(19.9000,-8.0000){\makebox(0,0){$\scriptsize{3}$}}%
%
\special{pn 8}%
\special{ar 1990 200 100 100 0.0000000 6.2831853}%
\put(19.9000,-2.0000){\makebox(0,0){$\scriptsize{1}$}}%
\end{picture}}%
\caption{}
\label{fig:fig2}
\end{figure}

Denote by $F(\mathcal{S}_{k_{1}, \ldots , k_{r}})$ the set of $\mathbb{C}$-valued functions on 
$\mathcal{S}_{k_{1}, \ldots , k_{r}}$. 
The backward generator $\mathcal{H}$ of the multi-species $q$-Boson system is the linear operator 
on $F(\mathcal{S}_{k_{1}, \ldots , k_{r}})$ defined by 
\begin{align*}
(\mathcal{H}h)(\vec{x}, \vec{\nu})=
\sum_{\begin{subarray}{c} (\vec{y}, \vec{\mu}) \in \mathcal{S}_{k_{1}, \ldots , k_{r}} \\ 
(\vec{y}, \vec{\mu})\not=(\vec{x}, \vec{\nu}) \end{subarray}}
q(\vec{x}, \vec{\nu} | \vec{y}, \vec{\mu})
\left\{h(\vec{y}, \vec{\mu})-h(\vec{x}, \vec{\nu})\right\} \qquad (h \in F(\mathcal{S}_{k_{1}, \ldots , k_{r}})), 
\end{align*}
where $q(\vec{x}, \vec{\nu} | \vec{y}, \vec{\mu})$ is the transition rate \eqref{eq:rate} from 
$(\vec{x}, \vec{\nu})$ to $(\vec{y}, \vec{\mu})$. 


\section{Algebraic construction of the eigenfunctions}\label{sec:eigen}

In this section we recall the construction of the eigenfunctions 
of the backward generator $\mathcal{H}$ following \cite{T}. 
There, the eigenfunctions are expressed in terms of the root system and the Weyl group of type $A_{k-1}$. 
Instead, here we represent them in a combinatorial way.

Let $U=\oplus_{a=1}^{r}\mathbb{C}u_{a}$ be the $r$-dimensional vector space. 
We denote by $(U^{\otimes k})_{k_{1}, \ldots , k_{r}}$ the subspace of $U^{\otimes k}$ 
spanned by the monomial vectors $u_{\nu_{1}} \otimes \cdots \otimes u_{\nu_{k}}$ with 
$(\nu_{1}, \ldots , \nu_{k}) \in I_{k_{1}, \ldots , k_{r}}$. 
For a function $\gamma$ on $L_{k}^{+}$ taking values in $(U^{\otimes k})_{k_{1}, \ldots , k_{r}}$, 
we define $\mathbb{C}$-valued functions $\gamma_{\vec{\nu}} \, (\vec{\nu} \in I_{k_{1}, \ldots , k_{r}})$
on $L_{k}^{+}$ by the following relation: 
\begin{align*}
\gamma(\vec{x})=\sum_{\vec{\nu} \in I_{k_{1}, \ldots , k_{r}}}q^{t(\vec{\nu})}\gamma_{\vec{\nu}}(\vec{x}) \,  
u_{\nu_{1}} \otimes \cdots \otimes u_{\nu_{k}} \qquad (\vec{x} \in L_{k}^{+}),  
\end{align*}
where $t(\vec{\nu})$ is the inversion number of $\vec{\nu}$: 
\begin{align}
t(\vec{\nu})=\#\{ (i, j) \, | \, 1\le i<j \le k \,\, \mathrm{and} \,\, \nu_{i}>\nu_{j} \}.  
\label{eq:def-inversion-number}
\end{align}

Denote by $\mathcal{F}(L_{k}^{+}, (U^{\otimes k})_{k_{1}, \ldots , k_{r}})$ 
the space of functions $\gamma: L_{k}^{+} \to (U^{\otimes k})_{k_{1}, \ldots , k_{r}}$ 
such that $\gamma_{\ldots , \nu_{i}, \nu_{i+1}, \ldots}=\gamma_{\ldots , \nu_{i+1}, \nu_{i}, \ldots}$ 
on the diagonal set $\{\vec{x} \in L_{k}^{+} \, | \, x_{i}=x_{i+1}\}$ for any $1\le i<k$.  
Then there exists an isomorphism 
$\varphi: F(\mathcal{S}_{k_{1}, \ldots , k_{r}}) \to
\mathcal{F}(L_{k}^{+}, (U^{\otimes k})_{k_{1}, \ldots , k_{r}})$ 
whose inverse is given by 
\begin{align}
\varphi^{-1}: 
\mathcal{F}(L_{k}^{+}, (U^{\otimes k})_{k_{1}, \ldots , k_{r}}) \to 
F(\mathcal{S}_{k_{1}, \ldots , k_{r}}), \qquad 
(\varphi^{-1} \gamma)(\vec{x}, \vec{\nu})=\gamma_{\vec{\nu}}(\vec{x}). 
\label{eq:def-Phi}
\end{align}

We define the linear operator $R \in \mathrm{End}(U^{\otimes 2})$ by 
\begin{align*}
R(u_{a}\otimes u_{b})=\left\{ 
\begin{array}{ll}
q\, u_{b} \otimes u_{a} & (a>b) \\
u_{a} \otimes u_{a} & (a=b) \\ 
(1-q^{2})u_{a} \otimes u_{b}+q\, u_{b}\otimes u_{a} & (a<b). 
\end{array}
\right. 
\end{align*}
For $1\le i<k$, we denote by $R_{i}$ 
the linear operator acting on the tensor product of the $i$-th and the $(i+1)$-th components  
of $U^{\otimes k}$ as the operator $R$. 
Then it holds that 
\begin{align*}
& 
(R_{i}-1)(R_{i}+q^{2})=0 \quad (1\le i <k), \\ 
& 
R_{i}R_{i+1}R_{i}=R_{i+1}R_{i}R_{i+1} \quad (1\le i \le k-2), \\  
& 
R_{i}R_{j}=R_{j}R_{i} \quad (|i-j|\ge 2). 
\end{align*}
Hence the operators $R_{i} \, (1\le i<k)$ define an action of 
the Hecke algebra of type $A_{k-1}$ \cite{J}.  
Note that the subspace $(U^{\otimes k})_{k_{1}, \ldots , k_{r}}$ is invariant under the action. 

Hereafter we set 
\begin{align}
f(z, w)=\frac{z-q^2 w}{z-w}, \qquad 
g(z, w)=-\frac{(1-q^{2})z}{z-w}=f(w, z)-1. 
\label{eq:def-fg}
\end{align} 
For $1\le i<k$ and $z, w \in \mathbb{C}$, set 
\begin{align}
Y_{i}(z, w)=\frac{1}{f(z, w)}\left(R_{i}+g(z, w)\right) \in \mathrm{End}(U^{\otimes k}).   
\label{eq:def-Y_{i}}
\end{align}
For $\tau \in \mathfrak{S}_{k}$ and $\vec{z}=(z_{1}, \ldots , z_{k}) \in \mathbb{C}^{k}$, 
the linear operator $\phi(\tau; \vec{z})$ acting on $U^{\otimes k}$ 
is uniquely determined from the following properties: 
\begin{align*}
\phi(1; \vec{z})=1, \quad 
\phi(\sigma_{i}\tau; \vec{z})=Y_{i}(z_{\tau^{-1}(i)}, z_{\tau^{-1}(i+1)})\phi(\tau; \vec{z}),  
\end{align*} 
where $\sigma_{i}=(i, i+1) \, (1\le i<k)$ is the transposition. 

\begin{lem}\label{lem:phi-decompose}
For $\tau, \tau' \in \mathfrak{S}_{k}$ and $\vec{z}=(z_{1}, \ldots , z_{k}) \in \mathbb{C}^{k}$, 
it holds that 
\begin{align*}
\phi(\tau'\tau; \vec{z})=\phi(\tau'; z_{\tau^{-1}(1)}, \ldots , z_{\tau^{-1}(k)}) \phi(\tau; \vec{z}).  
\end{align*}
\end{lem}

\begin{proof}
It follows from the properties which define the operator $\phi(\tau; \vec{z})$.  
\end{proof}

Now we write down the formula of the eigenfunctions. 

\begin{defn}
For $\vec{\mu}=(\mu_{1}, \ldots , \mu_{k}) \in I_{k_{1}, \ldots , k_{r}}$ and 
$\vec{z}=(z_{1}, \ldots , z_{k}) \in \mathbb{C}^{k}$, 
we define the function $h^{\vec{\mu}}_{\vec{z}}$ on $L_{k}^{+}$ taking values in $U^{\otimes k}$ by 
\begin{align*}
h_{\vec{z}}^{\vec{\mu}}(\vec{x})=\prod_{1\le i<j\le k}f(z_{i}, z_{j}) 
\sum_{\tau \in \mathfrak{S}_{k}}
\prod_{i=1}^{k}\left(\frac{z_{\tau^{-1}(i)}}{1+z_{\tau^{-1}(i)}}\right)^{x_{i}}\phi(\tau; \vec{z})
(u_{\mu_{1}} \otimes \cdots \otimes u_{\mu_{k}}). 
\end{align*} 
\end{defn}

\begin{thm}\label{thm:h-eigen}\cite{T}
The function $h^{\vec{\mu}}_{\vec{z}}$ belongs to $\mathcal{F}(L_{k}^{+}, (U^{\otimes k})_{k_{1}, \ldots , k_{r}})$ 
and $\varphi^{-1} h^{\vec{\mu}}_{\vec{z}}$ is an eigenfunction of the backward generator $\mathcal{H}$ 
with the eigenvalue $\sum_{i=1}^{k}z_{i}^{-1}$. 
\end{thm}

\begin{rem}
In \cite{T} the eigenfunctions are parametrized by $m \in U^{\otimes k}$ and 
$\lambda \in \left(\oplus_{i=1}^{k}\mathbb{R}v_{i}\right)^{*}$. 
They correspond to $\vec{\mu}$ and $\vec{z}$ respectively via the relations 
$m=u_{\mu_{1}} \otimes \cdots \otimes u_{\mu_{k}}$ and $e^{\lambda(v_{i})}=z_{i}/(1+z_{i})$.  
We also put the additional factor $\prod_{1\le i<j \le k}f(z_{i}, z_{j})$ to make 
the function $h_{\vec{z}}^{\vec{\mu}}$ have symmetry (see Proposition \ref{prop:symmetry} below). 
\end{rem}

\begin{rem}\label{rem:h-when-r=1}
When $r=1$, the vector space $U^{\otimes k}$ is one-dimensional and it holds that 
\begin{align*}
h_{\vec{z}}^{1^{k}}(\vec{x})=\sum_{\sigma \in \mathfrak{S}_{k}} 
\prod_{1\le i<j \le k}\!
\frac{z_{\sigma(i)}-q^2 z_{\sigma(j)}}{z_{\sigma(i)}-z_{\sigma(j)}}\,
\prod_{i=1}^{k}\left(\frac{z_{\sigma(i)}}{1+z_{\sigma(i)}}\right)^{x_{i}} \, 
(u_{1} \otimes \cdots \otimes u_{1}). 
\end{align*}
Thus we recover the eigenfunctions for the $q$-Boson system 
constructed in \cite{BCPS} by changing the parameters $z_{i} \mapsto -z_{i}^{-1}$. 
\end{rem}


\section{The algebra of $q$-deformed bosons and integrability}\label{sec:q-boson-alg}

\subsection{The $L$-operator}

We define the algebra $\mathcal{B}$ of multi-component $q$-deformed bosons to be 
the unital associative $\mathbb{C}$-algebra 
with the generators $\beta_{a}, \beta_{a}^{*}, q^{\pm N_{a}}  \, (1\le a \le r)$ satisfying 
\begin{align}
& \label{eq:commrel-q-boson}
q^{N_{a}}\beta_{a}=q^{-1}\beta_{a} q^{N_{a}}, \quad 
q^{N_{a}}\beta^{*}_{a}=q\,\beta_{a}^{*} q^{N_{a}}, \\ 
& 
\beta_{a}\beta_{a}^{*}=1-q^{2}q^{2N_{a}}, \quad 
\beta_{a}^{*}\beta_{a}=1-q^{2 N_{a}}
\nonumber 
\end{align}
and such that the elements $\beta_{a}, \beta_{a}^{*}, q^{\pm N_{a}}$ and 
$\beta_{b}, \beta_{b}^{*}, q^{\pm N_{b}}$ commute unless $a=b$. 
We abbreviate $\prod_{a=1}^{r}(q^{N_{a}})^{n_{a}}$ to $q^{\sum_{a=1}^{r}n_{a} N_{a}}$ for 
$n_{1}, \ldots , n_{r} \in \mathbb{Z}$. 

In \cite{GGW} a generalized $L$-operator of higher rank is introduced. 
We make use of it with parameters fixed (see Remark \ref{rem:L-op} below). 
The $L$-operator $L(z)=(L(z)_{ab})_{a, b=0}^{r} \in \mathcal{B} \otimes \mathrm{End}(\mathbb{C}^{r+1})$ 
is defined by 
\begin{align}
& \label{eq:def-L-operator}
L(z)_{00}=1+zq^{2\sum_{p=1}^{r}N_{p}}, \\ 
& 
L(z)_{0a}=\beta_{a}^{*}, \quad 
L(z)_{a0}=z\beta_{a}q^{2\sum_{p=a+1}^{r}N_{p}} \qquad (1 \le a \le r), 
\nonumber \\ 
& 
L(z)_{ab}=\left\{ 
\begin{array}{ll}
0 & (1 \le a<b \le r) \\ 
zq^{2\sum_{p=a+1}^{r}N_{p}} & (1 \le a=b \le r) \\ 
z\beta_{a}\beta_{b}^{*}q^{2\sum_{p=a+1}^{r}N_{p}} & (1\le b<a \le r). 
\end{array}
\right. 
\nonumber
\end{align}

We denote the matrix unit by  $E_{ab}=(\delta_{ai}\delta_{jb})_{i,j=0}^{r}$ and define  
$L(z)\otimes L(w) \in \mathcal{B}\otimes\mathrm{End}(\mathbb{C}^{r+1}\otimes\mathbb{C}^{r+1})$ by 
\begin{align*}
L(z)\otimes L(w)=\sum_{a,b=0}^{r}\sum_{c, d=0}^{r} 
\left(L(z)_{ab}L(w)_{cd}\right)\otimes E_{ab}\otimes E_{cd}.  
\end{align*}
Then the following relation holds. 

\begin{prop}\cite{GGW}
The $L$-operator satisfies the Yang-Baxter equation 
\begin{align}
\check{R}(z/w)[L(z)\otimes L(w)]=[L(w) \otimes L(z)] \check{R}(z/w), 
\label{eq:YBE-L}
\end{align}
where $\check{R}(z) \in \mathrm{End}(\mathbb{C}^{r+1}\otimes \mathbb{C}^{r+1})$ is given by 
\begin{align*}
\check{R}(z)&=(z-q^{2})\sum_{a=1}^{r}E_{aa}\otimes E_{aa}+
(1-q^{2})\sum_{1\le a<b \le r}\left(z E_{aa}\otimes E_{bb}+E_{bb}\otimes E_{aa}\right) \\ 
&+(z-1)\sum_{1\le a<b \le r}\left(E_{ab}\otimes E_{ba}+q^{2}E_{ba}\otimes E_{ab}\right). 
\end{align*}
\end{prop}

\begin{rem}\label{rem:L-op}
Our $L$-operator $L(z)$ is obtained from that defined in \cite{GGW} by setting 
$t=q^{-2}, \phi_{a}=\beta_{a}^{*}, \phi_{a}^{\dagger}=\beta_{a}, k_{a}=q^{2N_{a}}, u=-1$ and $v=0$. 
\end{rem}

\subsection{The monodromy matrix}

We introduce commutative copies $\mathcal{B}^{(i)} \, (i \in \mathbb{Z})$ of
the algebra $\mathcal{B}$. 
Denote the generators of $\mathcal{B}^{(i)}$ by 
$\beta_{a, i}, \beta^{*}_{a, i}$ and $q^{\pm N_{a, i}} \, (1\le a \le r)$. 
The elements $\beta_{a, i}, \beta^{*}_{a, i}, q^{\pm N_{a, i}}$ satisfy the relations \eqref{eq:commrel-q-boson} 
and commute with $\beta_{b, j}, \beta^{*}_{b, j}, q^{\pm N_{b, j}}$ such that $(a,i)\not=(b, j)$

We assign $\mathcal{B}^{(i)}$ to the $i$-th site on the one-dimensional lattice $\mathbb{Z}$, 
and define the $L$-operator $L^{(i)}(z)$ by \eqref{eq:def-L-operator} with 
$\beta_{a}, \beta^{*}_{a}, q^{\pm N_{a}}$ replaced by 
$\beta_{a, i}, \beta^{*}_{a, i}, q^{\pm N_{a, i}} \, (1\le a \le r)$.  

Let $M$ and $M'$ be integers such that $M'\le M$. 
Denote by $\mathcal{B}^{[M', M]}$ the algebra generated by 
$\mathcal{B}^{(i)} \, (M'\le i \le M)$. 
We define the monodromy matrix $\mathbb{T}^{[M', M]}(z)$ by 
\begin{align*}
\mathbb{T}^{[M', M]}(z)=L^{(M')}(z)L^{(M'+1)}(z)\cdots L^{(M-1)}(z)L^{(M)}(z),  
\end{align*}
which belongs to $\mathcal{B}^{[M', M]} \otimes \mathrm{End}(\mathbb{C}^{r+1})$. 
Denote its matrix element by $\mathbb{T}^{[M', M]}(z)_{ab} \, (0\le a, b \le r)$. 
{}From the Yang-Baxter equation \eqref{eq:YBE-L} we have 
\begin{align}
\check{R}(z/w)[\mathbb{T}^{[M', M]}(z)\otimes \mathbb{T}^{[M', M]}(w)]=
[\mathbb{T}^{[M', M]}(w) \otimes \mathbb{T}^{[M', M]}(z)] \check{R}(z/w),    
\label{eq:YBE-T}
\end{align}
where $\mathbb{T}^{[M', M]}(z)\otimes \mathbb{T}^{[M', M]}(w) \in 
\mathcal{B}^{[M', M]}\otimes\mathrm{End}(\mathbb{C}^{r+1}\otimes\mathbb{C}^{r+1})$ is defined by 
\begin{align*}
\mathbb{T}^{[M', M]}(z)\otimes \mathbb{T}^{[M', M]}(w)=\sum_{a,b=0}^{r}\sum_{c, d=0}^{r} 
\left(\mathbb{T}^{[M', M]}(z)_{ab}\mathbb{T}^{[M', M]}(w)_{cd}\right)\otimes E_{ab}\otimes E_{cd}.  
\end{align*}

Now set 
\begin{align*}
A^{[M', M]}(z)=\mathbb{T}^{[M', M]}(z)_{00}, \quad 
C_{a}^{[M', M]}(z)=\mathbb{T}^{[M', M]}(z)_{a0} \,\, (1\le a \le r). 
\end{align*}
The equality \eqref{eq:YBE-T} implies the following commutation relations: 

\begin{prop}\label{prop:commrel-AC}
We fix $M'$ and $M$, and abbreviate $X^{[M', M]}(z) \, (X=A, C_{a})$ to $X(z)$. 
It holds that 
\begin{align}
& 
A(z)A(w)=A(w)A(z), \quad C_{a}(z)C_{a}(w)=C_{a}(w)C_{a}(z) \quad (1\le a \le r), 
\nonumber \\ 
& 
C_{a}(z)A(w)=f(z, w)A(w)C_{a}(z)+g(z, w)A(z)C_{a}(w) \quad (1\le a \le r), 
\label{eq:commrel-AC}
\\ 
& 
q^{2}C_{b}(z)C_{a}(w)=f(z, w)C_{a}(w)C_{b}(z)-g(w, z)C_{a}(z)C_{b}(w) \quad (1\le b<a \le r), 
\label{eq:commrel-CC}
\end{align}
where $f(z, w)$ and $g(z, w)$ are defined by \eqref{eq:def-fg}. 
\end{prop}
Moreover the following recurrence relations hold. 
\begin{prop}\label{prop:recursion-AC}
Suppose that $M \ge 2$. Then it holds that 
\begin{align}
A^{[1, M]}(z)&=
(1+zq^{2\sum_{p=1}^{r}N_{p, 1}})A^{[2, M]}(z)+\sum_{p=1}^{r}\beta_{p, 1}^{*}C_{p}^{[2, M]}(z), 
\label{eq:recursion-A} \\ 
C_{a}^{[1, M]}(z)&=zq^{2\sum_{p=a+1}^{r}N_{p, 1}}\left\{ 
\beta_{a, 1}\left(A^{[2, M]}(z)+\sum_{p=1}^{a-1}\beta_{p, 1}^{*}C_{p}^{[2, M]}(z)\right)+C_{a}^{[2, M]}(z)
\right\}, 
\label{eq:recursion-C1} \\ 
C_{a}^{[1, M]}(z)&=C_{a}^{[1, M-1]}(z)\left(1+zq^{2\sum_{p=1}^{r}N_{p, M}}\right)+
z\sum_{b=1}^{r}\mathbb{T}_{a, b}^{[1, M-1]}(z)\beta_{b, M}q^{2\sum_{p=b+1}^{r}N_{p, M}}
\label{eq:recursion-C2}
\end{align}
for $1\le a \le r$. 
\end{prop}

\begin{proof}
Write down the matrix elements in the zeroth column of 
\begin{align*}
\mathbb{T}^{[1, M]}(z)=L^{(1)}(z)\mathbb{T}^{[2, M]}(z)=\mathbb{T}^{[1, M-1]}(z)L^{(M)}(z), 
\end{align*}
and we get the relations above. 
\end{proof}

\subsection{The integrability of the multi-species $q$-Boson model with periodic boundary}
\label{subsec:tranfer-matrix}

The Fock representation $\mathcal{F}$ of the algebra $\mathcal{B}$ is the vector space 
\begin{align*}
\mathcal{F}=\bigoplus_{m_{1}, \ldots m_{r} \in \mathbb{Z}_{\ge 0}}
\mathbb{C} | m_{1}, \ldots , m_{r} \rangle
\end{align*}
equipped with the left $\mathcal{B}$-module structure defined by 
\begin{align}
& \label{eq:action-Fock}
q^{N_{a}} | m_{1}, \ldots , m_{r} \rangle =q^{m_{a}}| m_{1}, \ldots , m_{r} \rangle, \\
& 
\beta_{a}^{*} | m_{1}, \ldots , m_{r} \rangle =| m_{1}, \ldots , m_{a}+1 , \ldots , m_{r} \rangle, 
\nonumber \\ 
&
\beta_{a} | m_{1}, \ldots , m_{r} \rangle =(1-q^{2 m_{a}})| m_{1}, \ldots , m_{a}-1 , \ldots , m_{r} \rangle. 
\nonumber 
\end{align}
The right action on its dual module 
\begin{align*}
\mathcal{F}^{*}=\bigoplus_{m_{1}, \ldots m _{r} \in \mathbb{Z}_{\ge 0}}
\mathbb{C} \langle m_{1}, \ldots , m_{r} | 
\end{align*}
is given by 
\begin{align*}
& 
\langle m_{1}, \ldots , m_{r} | q^{N_{a}}=q^{m_{a}}\langle m_{1}, \ldots , m_{r} |, \\ 
& 
\langle m_{1}, \ldots , m_{r} | \beta_{a}=(1-q^{2(m_{a}+1)})\langle m_{1}, \ldots , m_{a}+1, \ldots, m_{r} |, \\
&
\langle m_{1}, \ldots , m_{r} | \beta_{a}^{*}=
\left\{ 
\begin{array}{ll}
0 & (m_{a}=0) \\ \langle m_{1}, \ldots , m_{a}-1, \ldots , m_{r} | & (m_{a}\ge 1). 
\end{array}
\right.
\end{align*}

Let $M$ be a positive integer. 
We introduce $M$ copies $\mathcal{F}^{(i)} \, (1\le i \le M)$ of the vector space $\mathcal{F}$ and set 
$\mathcal{F}^{[1, M]}=\mathcal{F}^{(1)} \otimes \cdots \otimes \mathcal{F}^{(M)}$. 
The algebra $\mathcal{B}^{[1, M]}$ acts on $\mathcal{F}^{[1, M]}$ in such a way that 
$\mathcal{B}^{(i)}$ acts on the $i$-th component $\mathcal{F}^{(i)}$ as given by \eqref{eq:action-Fock}. 
Then we define the transfer matrix $\tau(z) \in \mathrm{End}(\mathcal{F}^{[1,M]})$ by 
\begin{align*}
\tau(z)=\mathrm{tr}_{\mathbb{C}^{r+1}}(\mathbb{T}^{[1, M]}(z))=
\sum_{a=0}^{r}\mathbb{T}^{[1, M]}(z)_{aa}. 
\end{align*} 
Note that $\tau(z)$ is a polynomial in $z$ of degree $M$. 
Now define the operators $H_{n} \, (1\le n \le M)$ by the expansion $\tau(z)=\sum_{n=0}^{M}H_{n}z^{n}$. 
The relation \eqref{eq:YBE-T} implies that
\begin{align*}
\check{R}(z/w)[\mathbb{T}^{[1, M]}(z)\otimes \mathbb{T}^{[1, M]}(w)]\check{R}(z/w)^{-1}=
\mathbb{T}^{[1, M]}(w) \otimes \mathbb{T}^{[1, M]}(z).     
\end{align*}
Taking the trace over $\mathbb{C}^{r+1}\otimes \mathbb{C}^{r+1}$, we get 
\begin{align*}
\tau(z)\tau(w)=\tau(w)\tau(z).  
\end{align*}
Hence $H_{n} \, (0\le n \le M)$ commutes with each other. 

By direct calculation we see that $H_{0}=1$ and 
\begin{align*}
H_{1}=M+\sum_{i=1}^{M}h_{i-1, i}, \quad 
h_{i-1, i}=\sum_{a=1}^{r}
(\beta_{a, i-1}^{*}-\beta_{a, i}^{*})\beta_{a, i}q^{2\sum_{p=a+1}^{r}N_{a, i}},   
\end{align*} 
where the index $i$ should read modulo $M$. 
Note that 
\begin{align*}
& 
h_{12}(|m_{1}, \ldots , m_{r}\rangle \otimes |n_{1}, \ldots , n_{r}\rangle)=
\sum_{a=1}^{r}q^{2\sum_{p=a+1}^{r}n_{p}}(1-q^{2n_{a}}) \\ 
&\qquad {}\times\left( 
|\ldots , m_{a}+1, \ldots \rangle \otimes |\ldots, n_{a}-1, \ldots \rangle-
|\ldots , m_{a}, \ldots \rangle \otimes |\ldots, n_{a}, \ldots \rangle
\right). 
\end{align*}

Let $\mathcal{S}_{M}^{\mathrm{per}}$ be the set of configurations of bosonic particles with $r$ species on 
the one-dimensional lattice of length $M$ with periodic boundary.   
We assign each configuration $\mathbf{s}$ to the vector 
$|\mathbf{s}\rangle=\otimes_{i=1}^{M}(|n_{1}^{(i)}, \ldots , n_{r}^{(i)}\rangle)$ of $\mathcal{F}^{[1, M]}$, 
where $n_{a}^{(i)}$ is 
the number of particles with color $a$ at the $i$-th site in the configuration $\mathbf{s}$. 
Then we see that 
\begin{align*}
(H_{1}-M)|\mathbf{s}\rangle=(1-q^{2})
\sum_{\mathbf{s}' \in \mathcal{S}_{M}^{\mathrm{per}}\setminus\{\mathbf{s}\}}
q(\mathbf{s}|\mathbf{s}')(|\mathbf{s}'\rangle-|\mathbf{s}\rangle), 
\end{align*}
where $q(\mathbf{s}|\mathbf{s}')$ is the transition rate of the multi-species $q$-Boson system from 
$\mathbf{s}$ to $\mathbf{s}'$ given by \eqref{eq:rate}. 
Therefore $H_{1}$ gives the transition rate matrix of the system with periodic boundary. 
The operators $H_{n} \, (2\le n \le M)$ generated from the transfer matrix commute with it. 
Thus we may regard the multi-species $q$-Boson system as integrable. 


\section{Main theorem and its proof}\label{sec:main}

\subsection{Main theorem}

We now turn back to the system on the infinite lattice $\mathbb{Z}$. 
We introduce countably many copies $\mathcal{B}^{(i)}$ and $\mathcal{F}^{(i)} \, (i \in \mathbb{Z})$ 
of the algebra $\mathcal{B}$ and its Fock representation $\mathcal{F}$, respectively, 
and assign  $\mathcal{B}^{(i)}$ and $\mathcal{F}^{(i)}$ to the $i$-th site of the lattice $\mathbb{Z}$. 
For two integers $M'$ and $M$ such that $M'\le M$, 
we denote by $\mathcal{B}^{[M', M]}$ the algebra generated by $\mathcal{B}^{(i)} \, (M'\le i\le M)$ 
and set 
$\mathcal{F}^{[M', M]}=\mathcal{F}^{(M')} \otimes \mathcal{F}^{(M'+1)} \otimes \cdots \otimes \mathcal{F}^{(M)}$. 
Set $|\mathbf{0}\rangle=|0, \ldots , 0 \rangle  \in \mathcal{F}$ and 
$\langle \mathbf{0} |=\langle 0, \ldots , 0 | \in \mathcal{F}^{*}$. 
We denote the vacuum vector and its dual on the interval $[M', M]$ by
\begin{align*}
& 
|\mathrm{vac}\rangle_{[M', M]}=|\mathbf{0} \rangle \otimes 
|\mathbf{0} \rangle \otimes \cdots \otimes |\mathbf{0} \rangle \in \mathcal{F}^{[M', M]},  \\ 
& 
{}_{[M', M]}\langle \mathrm{vac} |=\langle \mathbf{0} | \otimes 
\langle \mathbf{0} | \otimes \cdots \otimes \langle \mathbf{0} | 
\in (\mathcal{F^{*}})^{[M', M]}. 
\end{align*}
For simplicity we write 
\begin{align*}
\langle X \rangle_{[M', M]}={}_{[M, M']}\langle \mathrm{vac} | X 
|\mathrm{vac}\rangle_{[M', M]} \qquad (X \in \mathcal{B}^{[M', M]}). 
\end{align*}

\begin{prop}
Suppose that $(x_{1}, \ldots , x_{k}) \in L_{k}^{+}$ 
and $(\mu_{1}, \ldots , \mu_{k}), (\nu_{1}, \ldots , \nu_{k}) \in I_{k_{1}, \ldots , k_{r}}$. 
Then the following quantity does not depend on the choice of $M$ and $M'$ such that $M \ge x_{1}$ and 
$x_{k} \ge M'$: 
\begin{align}
\prod_{i=1}^{k}\frac{z_{i}^{M'-1}}{(1+z_{i})^{M}} 
\langle 
\prod_{1\le i \le k}^{\curvearrowright}C_{\mu_{i}}^{[M, M']}(z_{i})
\prod_{1\le i \le k}\beta_{\nu_{i}, x_{i}}^{*}
\rangle_{[M', M]}.  
\label{eq:expectation-shift}
\end{align} 
\end{prop}

\begin{proof}
Let $M \ge 2$ and $X \in \mathcal{B}^{[1, M]}$.  
{}From Proposition \ref{prop:recursion-AC} it holds that 
\begin{align*}
& 
C_{a}^{[1, M+1]}(z)X|\mathrm{vac}\rangle_{[1, M+1]}=(1+z)
C_{a}^{[1, M]}(z)X|\mathrm{vac}\rangle_{[1, M]}, \\ 
& 
C_{a}^{[0, M]}(z)X|\mathrm{vac}\rangle_{[0, M]}=z\,
C_{a}^{[1, M]}(z)X|\mathrm{vac}\rangle_{[1, M]}. 
\end{align*}
for $1\le a \le r$. 
The properties above imply that the value \eqref{eq:expectation-shift} is equal to that with 
$M'$ and $M$ replaced by $x_{k}$ and $x_{1}$, respectively. 
\end{proof}

Now we introduce the following function. 
\begin{defn}
For $\vec{\mu}=(\mu_{1}, \ldots , \mu_{k}) \in I_{k_{1}, \ldots , k_{r}}$  and 
$\vec{z}=(z_{1}, \ldots , z_{k}) \in \mathbb{C}^{k}$, 
we define the function $\psi_{\vec{z}}^{\vec{\mu}}$ on $L_{k}^{+}$ taking values in 
$(U^{\otimes k})_{k_{1}, \ldots , k_{r}}$ by 
\begin{align*}
\psi_{\vec{z}}^{\vec{\mu}}(\vec{x})=(1-q^{2})^{-k}
\sum_{\vec{\nu} \in I_{k_{1}, \ldots , k_{r}}}& 
\prod_{i=1}^{k}\frac{z_{i}^{M'-1}}{(1+z_{i})^{M}} 
\langle 
\prod_{1\le i \le k}^{\curvearrowright}C_{\mu_{i}}^{[M', M]}(z_{i})
\prod_{1\le i \le k}\beta_{\nu_{i}, x_{i}}^{*}
\rangle_{[M', M]} \\ 
&\times 
q^{t(\vec{\nu})-t(\vec{\mu})} u_{\nu_{1}} \otimes \cdots \otimes u_{\nu_{k}},    
\end{align*} 
where $M$ and $M'$ are integers such that $M \ge x_{1}$ and $x_{k} \ge M'$, 
and $t(\cdot)$ is the inversion number \eqref{eq:def-inversion-number}. 
\end{defn}

\begin{ex}\label{ex:psi-when-r=1}
Let us consider the case of $r=1$. 
In the following we drop 
the subscript $1$ on $\beta_{1, i}, \beta_{1, i}^{*}, q^{N_{1, i}}$ and $C_{1}^{[M', M]}(z)$ for brevity. 
The $L$-operator is in the form of $2\times 2$ matrix: 
\begin{align}
L^{(i)}(z)=\begin{pmatrix} 1+zq^{2N_{i}} & \beta_{i}^{*} \\ z\beta_{i} & z \end{pmatrix}.   
\label{eq:L-op-r=1}
\end{align} 
Then the operator $C^{[M', M]}(z)$ is defined by 
\begin{align}
C^{[M', M]}(z)=\begin{pmatrix} 0 & 1 \end{pmatrix}
L^{(M')}(z)L^{(M'+1)}(z)\cdots L^{(M)}(z)
\begin{pmatrix} 1 \\ 0 \end{pmatrix}.  
\label{eq:def-C-r=1}
\end{align}
The function $\psi^{\vec{\mu}}_{\vec{z}}$ is defined only for $\vec{\mu}=(1, \ldots , 1)$ and 
we see that 
\begin{align*}
\psi_{\vec{z}}^{1^{k}}(\vec{x})=(1-q^{2})^{-k}\prod_{i=1}^{k}\frac{z_{i}^{M'-1}}{(1+z_{i})^{M}} 
\langle \prod_{1\le i \le k}C^{[M', M]}(z_{i}) 
\prod_{1\le i \le k}\beta_{x_{i}}^{*} \rangle_{[M', M]}\, 
(u_{1}\otimes \cdots \otimes u_{1}). 
\end{align*}
Note that $[C^{[M', M]}(z), C^{[M', M]}(w)]=0$ because of Proposition \ref{prop:commrel-AC}.  
Hence the function $\psi^{1^{k}}_{\vec{z}}$ is symmetric in $\vec{z}=(z_{1}, \ldots , z_{k})$. 
\end{ex}

Our main theorem is as follows. 

\begin{thm}\label{thm:main}
It holds that $h^{\vec{\mu}}_{\vec{z}}=\psi^{\vec{\mu}}_{\vec{z}}$ 
for any $\vec{\mu} \in I_{k_{1}, \ldots , k_{r}}$ and $\vec{z} \in \mathbb{C}^{k}$. 
\end{thm}

We prove the theorem in the following subsections. 
Here we note that it gives us an expression of the eigenfunctions for the multi-species 
$q$-Boson system in terms of the $q$-deformed bosonic operators: 

\begin{cor}
For $\vec{\mu}=(\mu_{1}, \ldots , \mu_{k}) \in I_{k_{1}, \ldots , k_{r}}$  and 
$\vec{z}=(z_{1}, \ldots , z_{k}) \in \mathbb{C}^{k}$, we define the function 
$E_{\vec{z}}^{\vec{\mu}}$ on $\mathcal{S}_{k_{1}, \ldots , k_{r}}$ by 
\begin{align}
E_{\vec{z}}^{\vec{\mu}}(\vec{x}, \vec{\nu})=
\lim_{\begin{subarray}{c} M \to \infty \\ M' \to -\infty \end{subarray}}
\prod_{i=1}^{k}\frac{z_{i}^{M'-1}}{(1+z_{i})^{M}} 
\langle 
\prod_{1\le i \le k}^{\curvearrowright}C_{\mu_{i}}^{[M', M]}(z_{i})
\prod_{1\le i \le k}\beta_{\nu_{i}, x_{i}}^{*}
\rangle_{[M', M]}.   
\label{eq:cor-main}
\end{align} 
Then $E_{\vec{z}}^{\vec{\mu}}$ is an eigenfunction of the backward generator $\mathcal{H}$ 
with the eigenvalue $\sum_{i=1}^{k}z_{i}^{-1}$. 
\end{cor}

\begin{proof}
It follows from Theorem \ref{thm:h-eigen}, Theorem \ref{thm:main} and 
the equality $E_{\vec{z}}^{\vec{\mu}}=(1-q^{2})^{k}q^{t(\vec{\mu})}(\varphi^{-1} \psi_{\vec{z}}^{\vec{\mu}})$, 
where $\varphi^{-1}$ is the isomorphism defined by \eqref{eq:def-Phi}.  
\end{proof}

\begin{rem}
We proved above that $E_{\vec{z}}^{\vec{\mu}}$ is an eigenfunction by 
showing that it is equal to $\varphi^{-1} h_{\vec{z}}^{\vec{\mu}}$ up to a constant factor. 
The statement might be proved directly from the expression \eqref{eq:cor-main}. 
Here we do not discuss this point further. 
\end{rem}

\begin{rem}
If $r=1$, then the equality $h_{\vec{z}}^{\vec{\mu}}=\psi_{\vec{z}}^{\vec{\mu}}$ in Theorem \ref{thm:main} 
can be derived from the combinatorial formula for a family of symmetric rational functions 
due to Borodin \cite{Bor} as follows. 

Note that $h_{\vec{z}}^{\vec{\mu}}$ and $\psi_{\vec{z}}^{\vec{\mu}}$ are defined only for 
$\vec{\mu}=1^{k}$ if $r=1$. 
Hence it suffices to prove the following equality for any 
$\vec{z} \in \mathbb{C}^{k}$ and $\vec{x} \in L_{k}^{+}$: 
\begin{align}
& \label{eq:main-r=1}
\sum_{\sigma \in \mathfrak{S}_{k}} 
\prod_{1\le i<j \le k}\!
\frac{z_{\sigma(i)}-q^2 z_{\sigma(j)}}{z_{\sigma(i)}-z_{\sigma(j)}}\,
\prod_{i=1}^{k}\left(\frac{z_{\sigma(i)}}{1+z_{\sigma(i)}}\right)^{x_{i}} \\ 
&=(1-q^{2})^{-k}\prod_{i=1}^{k}\left(z_{i}^{M'-1}(1+z_{i})^{-M} \right)
\langle \prod_{1\le i \le k}C^{[M', M]}(z_{i}) 
\prod_{1\le i \le k}\beta_{x_{i}}^{*} \rangle_{[M', M]} 
\nonumber 
\end{align}
(see Remark \ref{rem:h-when-r=1} and Example \ref{ex:psi-when-r=1}). 
Moreover we can assume that $x_{k} \ge 0$ and $M'=0$ 
because of Proposition \ref{prop:shift} which will be proved 
in the next subsection. 

Following \cite{Bor} we introduce the function\footnote{
Here we replaced the parameter $q$ in \cite{Bor} with $q^{2}$.}   
\begin{align*}
F_{\vec{x}}(u_{1}, \ldots , u_{k})=\frac{(1-q^2)^{k}}{\prod_{i=1}^{k}(1-su_{i})}
\sum_{\sigma \in \mathcal{S}_{k}}
\prod_{1\le i<j \le k}\frac{u_{\sigma(i)}-q^{2}u_{\sigma(j)}}{u_{\sigma(i)}-u_{\sigma(j)}} 
\prod_{i=1}^{k}\left(\frac{u_{i}-s}{1-su_{i}}\right)^{x_{i}}  
\end{align*}
for $\vec{x} \in L_{k}^{+}$ such that $x_{k}\ge 0$,  
where $s$ is a parameter.   
The combinatorial formula proved in \cite{Bor} represents the function $F_{\vec{x}}$
as a partition function for path ensembles in the square grid. 
It can be rewritten in terms of the $q$-deformed bosons as follows. 

Define the operator 
$\tilde{L}(u; s)=(\tilde{L}(u; s)_{ab})_{a, b=0}^{1} \in \mathcal{B}\otimes \mathrm{End}(\mathbb{C}^{2})$ by  
\begin{align*}
\tilde{L}(u; s)=\frac{1}{1-su}
\begin{pmatrix} 1-suq^{2N} & u\beta^{*}(1-s^2 q^{2N}) \\ \beta & u-sq^{2N} \end{pmatrix}    
\end{align*}
and $\tilde{L}^{(i)}(u; s) \, (i \in \mathbb{Z})$ 
in the same way as before. 
Then the vertex weight $w_{u}(i_{1}, j_{i}; i_{2}, j_{2})$ defined in \cite{Bor} is realized as 
$w_{u}(i_{1}, j_{1}; i_{2}, j_{2})=\langle i_{1} | \tilde{L}(u; s)_{j_{1}j_{2}} | i_{2} \rangle$
for $i_{1}, i_{2} \in \mathbb{Z}_{\ge 0}$ and $j_{1}, j_{2} \in \{0, 1\}$. 
This correspondence implies that the combinatorial formula proved in \cite{Bor} is rewritten as 
\begin{align}
F_{\vec{x}}(u_{1}, \ldots , u_{k})=\langle \prod_{1\le i \le k}\tilde{C}^{[0, M]}(u_{i}; s) 
\prod_{1\le i \le k}\beta_{x_{i}}^{*} \rangle_{[0, M]},  
\label{eq:borodin}
\end{align}
where $M$ is an arbitrary integer such that $M \ge x_{1}$ and 
$\tilde{C}^{[0, M]}(u; s)$ is the operator acting on $\mathcal{F}^{[0, M]}$ 
defined by the right hand side of \eqref{eq:def-C-r=1} 
with $M'=0$ and $L^{(i)}(z)$ replaced by $\tilde{L}^{(i)}(u; s)$ $(0\le i \le M)$. 

Now we introduce the $s$-deformed $L$-operator 
\begin{align*}
L^{(i)}(z; s)=\begin{pmatrix} 1+zq^{2N_{i}} & \beta_{i}^{*}(1-s^{2}q^{2N_{i}}) \\ 
z\beta_{i} & z+s^{2}q^{2N_{i}} \end{pmatrix} 
\quad (0\le i \le M).    
\end{align*} 
Note that it is equal to the $L$-operator \eqref{eq:L-op-r=1} at $s=0$. 
For a parameter $\alpha \in \mathbb{C}$ we define the linear operator $K(\alpha)$ 
acting on the Fock space $\mathcal{F}$ by 
$K(\alpha)|m\rangle =\alpha^{m}|m\rangle \, (m \in \mathbb{Z}_{\ge 0})$. 
We denote by $K^{(i)}(\alpha)$ the linear operator acting on 
the $i$-component $\mathcal{F}^{(i)}$ of $\mathcal{F}^{[0, M]}$ as $K(\alpha)$. 
Then it holds that 
\begin{align*}
\left[L^{(i)}(z; s), \, \begin{pmatrix} 1 & 0 \\ 0 & \alpha\end{pmatrix}\!K^{(i)}(\alpha)\right]=0
\end{align*}
in $\mathcal{B}^{[0, M]}\otimes\mathrm{End}(\mathbb{C}^{2})$. 
Using this property and the relation 
\begin{align*}
\tilde{L}^{(i)}(u; s)=\frac{1}{1-su}
\begin{pmatrix} 1 & 0 \\ 0 & (-su)^{-1} \end{pmatrix} L^{(i)}(-su; s)
\begin{pmatrix} 1 & 0 \\ 0 & u \end{pmatrix}, 
\end{align*}
we find that 
\begin{align*}
\tilde{C}^{[0, M]}(u; s)&=(-su)^{-1}(-s(1-su))^{-(M+1)} 
\prod_{1\le i \le k}K^{(i)}((-s)^{-(M+1-i)}) \\ 
&\times 
C^{[0, M]}(-su; s)
\prod_{1\le i \le k}K^{(i)}((-s)^{M+1-i}),  
\end{align*}
where $C^{[0, M]}(z; s)$ is the operator defined by the right hand side of \eqref{eq:def-C-r=1} 
with $M'=0$ and $L^{(i)}(z)$ replaced by ${L}^{(i)}(z; s)$ $(0\le i \le M)$. 
Substitute this expression of $\tilde{C}^{[0, M]}(u; s)$ into \eqref{eq:borodin} 
and change the variables $u_{i}$ to $-z_{i}/s$ $(0 \le i \le k)$. 
Then we see that 
\begin{align}
F_{\vec{x}}(-z_{1}/s, \ldots , -z_{k}/s)=
\frac{(-s)^{\sum_{i=1}^{k}x_{i}}}{\prod_{i=1}^{k}z_{i}(1+z_{i})^{M+1}} 
\langle \prod_{1\le i \le k}C^{[0, M]}(u_{i}; s) 
\prod_{1\le i \le k}\beta_{x_{i}}^{*} \rangle_{[0, M]}.  
\label{eq:compare-borodin}
\end{align}

Now we can obtain the formula \eqref{eq:main-r=1} as follows. 
{}From the definition of $F_{\vec{x}}$,  
the left hand side of \eqref{eq:main-r=1} is equal to the limit 
\begin{align*}
\lim_{s \to 0}\left\{
(1-q^{2})^{-k}(-s)^{-\sum_{i=1}^{k}x_{i}}
\prod_{i=1}^{k}(1+z_{i}) 
\, F_{\vec{x}}(-z_{1}/s, \ldots , -z_{k}/s)
\right\}. 
\end{align*}
Using the formula \eqref{eq:compare-borodin} and 
$C^{[0, M]}(z; 0)=C^{[0, M]}(z)$, 
we find that the limit is equal to the right hand side of \eqref{eq:main-r=1}.  
\end{rem}

\begin{rem}
In \cite{MS} Motegi and Sakai prove a formula 
which represents the $\beta$-Grothendieck polynomial \cite{FK}  
as the wavefunction of the non-Hermitian phase model\footnote{
Here $\beta$ is a parameter and not the $q$-deformed boson.}. 
At the $K$-theoretical point $\beta=-1$, 
their formula is equivalent to the equality \eqref{eq:main-r=1} with $q=0$.  
\end{rem}

\subsection{Common properties of $h_{\vec{z}}^{\vec{\mu}}$ and $\psi_{\vec{z}}^{\vec{\mu}}$}

The proof of Theorem \ref{thm:main} will consist of establishing that 
the functions $h_{\vec{z}}^{\vec{\mu}}$ and $\psi_{\vec{z}}^{\vec{\mu}}$ share three common properties. 
First we show two of them. 

\begin{prop}\label{prop:shift}
Both $h^{\vec{\mu}}_{\vec{z}}$ and $\psi^{\vec{\mu}}_{\vec{z}}$ have 
the shift invariance 
\begin{align*}
F_{\vec{z}}^{\vec{\mu}}(x_{1}+c, \ldots , x_{k}+c)=
\prod_{i=1}^{k}\left(\frac{z_{i}}{1+z_{i}}\right)^{c} 
F_{\vec{z}}^{\vec{\mu}}(x_{1}, \ldots , x_{k}) \qquad 
(\vec{x} \in L_{+}, c \in \mathbb{Z}). 
\end{align*}
\end{prop}

\begin{proof}
It follows from the definition of $h^{\vec{\mu}}_{\vec{z}}$ and $\psi^{\vec{\mu}}_{\vec{z}}$. 
\end{proof}

\begin{prop}\label{prop:symmetry}
Both $h^{\vec{\mu}}_{\vec{z}}$ and $\psi^{\vec{\mu}}_{\vec{z}}$ satisfy 
the following symmetry conditions. 
\begin{enumerate}
 \item If $\vec{\mu} \in I_{k_{1}, \ldots , k_{r}}$ satisfies $\mu_{i}=\mu_{i+1}$, then 
$F^{\ldots , \mu_{i}, \mu_{i+1}, \ldots}_{\ldots , z_{i}, z_{i+1}, \ldots}=
F^{\ldots , \mu_{i}, \mu_{i+1}, \ldots}_{\ldots , z_{i+1}, z_{i}, \ldots}$. 
 \item If $\vec{\mu} \in I_{k_{1}, \ldots , k_{r}}$ satisfies $\mu_{i}<\mu_{i+1}$, then 
\begin{align*}
q \, F_{\ldots , z_{i}, z_{i+1}, \ldots}^{\ldots , \mu_{i}, \mu_{i+1}, \ldots}=
f(z_{i}, z_{i+1})F_{\ldots , z_{i+1}, z_{i}, \ldots}^{\ldots , \mu_{i+1}, \mu_{i}, \ldots}-
g(z_{i+1}, z_{i})F_{\ldots , z_{i}, z_{i+1}, \ldots}^{\ldots , \mu_{i+1}, \mu_{i}, \ldots},  
\end{align*}
where $f(z, w)$ and $g(z, w)$ are defined by \eqref{eq:def-fg}. 
\end{enumerate} 
\end{prop}

\begin{proof}
It follows from Lemma \ref{lem:phi-decompose} and Proposition \ref{prop:commrel-AC}. 
\end{proof}

Proposition \ref{prop:shift} and Proposition \ref{prop:symmetry} imply that, 
in order to prove Theorem \ref{thm:main},  
it suffices to show $h^{\vec{\mu}}_{\vec{z}}(x)=\psi^{\vec{\mu}}_{\vec{z}}(x)$ for 
$\vec{\mu}=(r^{k_{r}}, \ldots , 1^{k_{1}})$ and 
$x \in L_{k}^{+}$ such that $x_{k}=1$.  
Note that 
\begin{align*}
h_{z}^{a}(1)=\frac{z}{1+z} u_{a}=\psi_{z}^{a}(1)  
\end{align*}
for $z \in \mathbb{C}$ and $1\le a \le r$. 
Hence Theorem \ref{thm:main} follows from the recurrence relation below. 

\begin{thm}\label{thm:recursion}
Suppose that $x_{1}\ge \cdots \ge x_{k-1} \ge 1$ and $\vec{z} \in \mathbb{C}^{k}$. 
Both $h^{\vec{\mu}}_{\vec{z}}$ and $\psi^{\vec{\mu}}_{\vec{z}}$ satisfy the following relation. 
\begin{align}
& \label{eq:recursion} 
F_{\vec{z}}^{r^{k_{r}}, \ldots , 1^{k_{1}}}(x_{1}, \ldots , x_{k-1}, 1)=
\sum_{a=1}^{r}q^{\sum_{p=1}^{a-1}k_{p}} \\ 
& {}\times 
\sum_{t=0}^{r-a}\sum_{r\ge p(t)>\cdots >p(1)>p(0)=a} 
\sum_{\ell(t) \in J_{p(t)}} \cdots \sum_{\ell(0) \in J_{p(0)}} 
\frac{z_{\ell(t)}}{1+z_{\ell(t)}} 
\prod_{s=1}^{t}g(z_{\ell(s)}, z_{\ell(s-1)}) 
\nonumber \\ 
& \quad {}\times 
\prod_{s=0}^{t}(
\prod_{\begin{subarray}{c} i \in J_{p(s+1)-1} \cup \cdots \cup J_{p(s)} \\ i\not=\ell(s) \end{subarray}}
f(z_{i}, z_{\ell(s)})) \, 
F_{\vec{z}(\ell(t), \ldots , \ell(0))}^{r^{k_{r}}, \ldots , a^{k_{a}-1}, \ldots , 1^{k_{1}}}
(x_{1}, \ldots , x_{k-1}) \otimes u_{a},  
\nonumber
\end{align}
where we set $J_{p(t+1)-1}=J_{r}$ by agreement in the rightmost product, $J_{p} \, (1\le p \le r)$ is given by 
\begin{align}
J_{p}=\{ i \in \mathbb{Z} \, | \, 
k_{r}+\cdots +k_{p+1}<i\le k_{r}+\cdots +k_{p} \}, 
\label{eq:def-J_{p}}
\end{align}
and $\vec{z}(\ell(t), \ldots , \ell(0)) \in \mathbb{C}^{k-1}$ is defined by 
\begin{align}
\vec{z}(\ell(t),  \ldots , \ell(0))_{i}=\left\{
\begin{array}{ll}
z_{i} & (1\le i <\ell(0), i\not=\ell(t), \ldots , \ell(1)) \\ 
z_{\ell(s-1)} & (i=\ell(s), 1\le s \le t) \\ 
z_{i+1} & (\ell(0)\le i \le k-1). 
\end{array}
\right.  
\label{eq:def-z-ell}
\end{align}
\end{thm}

\begin{ex}
The $r=1$ case of the recurrence relation \eqref{eq:recursion} reads as 
\begin{align*}
F^{1^{k}}_{\vec{z}}(x_{1}, \ldots , x_{k-1}, 1)=\sum_{\ell=1}^{k}\frac{z_{\ell}}{1+z_{\ell}}
\prod_{0<i\le k \atop i\not=\ell}f(z_{i}, z_{\ell}) 
F^{1^{k-1}}_{\vec{z}(\ell)}(x_{1}, \ldots , x_{k-1}) \otimes u_{1},  
\end{align*}  
where 
$\vec{z}(\ell)=(z_{1}, \ldots , z_{\ell-1}, z_{\ell+1}, \ldots , z_{k})$. 
In the case of $r=2$, it reads as follows: 
\begin{align*}
&
F^{2^{k_{2}}, 1^{k_{1}}}_{\vec{z}}(x_{1}, \ldots , x_{k-1}, 1) \\ 
&=\sum_{k_{2}<\ell\le k}\frac{z_{\ell}}{1+z_{\ell}} 
\prod_{0<i \le k \atop i\not=\ell}f(z_{i}, z_{\ell}) 
F^{2^{k_{2}}, 1^{k_{1}-1}}_{\vec{z}(\ell)}(x_{1}, \ldots , x_{k-1}) \otimes u_{1} \\ 
&+\sum_{0<\ell_{1}\le k_{2} \atop k_{2}<\ell_{0}\le k}
\frac{z_{\ell_{1}}}{1+z_{\ell_{1}}}g(z_{\ell_{1}}, z_{\ell_{0}}) 
\prod_{0< i \le k_{2} \atop i\not=\ell_{1}}f(z_{i}, z_{\ell_{1}}) 
\prod_{k_{2}< i \le k \atop i\not=\ell_{0}}f(z_{i}, z_{\ell_{0}}) 
F^{2^{k_{2}}, 1^{k_{1}-1}}_{\vec{z}(\ell_{1}, \ell_{0})}(x_{1}, \ldots , x_{k-1}) \otimes u_{1}  \\ 
&+q^{k_{1}}\sum_{0<\ell\le k_{2}}\frac{z_{\ell}}{1+z_{\ell}} 
\prod_{0<i \le k_{2} \atop i\not=\ell}f(z_{i}, z_{\ell}) 
F^{2^{k_{2}-1}, 1^{k_{1}}}_{\vec{z}(\ell)}(x_{1}, \ldots , x_{k-1}) \otimes u_{2},  
\end{align*}
where 
\begin{align*}
\vec{z}(\ell_{1}, \ell_{0})=(z_{1}, \ldots , \underset{\hbox{\scriptsize $\ell_{1}$-th}}{z_{\ell_{0}}}, 
\ldots , z_{\ell_{0}-1}, z_{\ell_{0}+1}, \ldots , z_{k}).  
\end{align*}
\end{ex}

In the rest of this paper we prove the recurrence relation \eqref{eq:recursion} for 
$h_{\vec{z}}^{\vec{\mu}}$ and $\psi_{\vec{z}}^{\vec{\mu}}$ successively.

\subsection{Proof for $h_{\vec{z}}^{\vec{\mu}}$}

First we prove the relation \eqref{eq:recursion} for 
$h^{\vec{\mu}}_{\vec{z}}$. 
For that purpose we will show some relations on the space $U_{0} \otimes (U^{\otimes m})$, 
where $m$ is a positive integer and $U_{0}$ is a copy of $U$ which we regard as the zeroth component.  
Then we set $R_{i}=(1^{\otimes i}) \otimes R \otimes (1^{\otimes (m-i-1)}) \in 
\mathrm{End}(U_{0} \otimes (U^{\otimes m}))$ 
and define $Y_{i}(z, w)$ by \eqref{eq:def-Y_{i}} for $0\le i<m$. 

Let $m$ be a positive integer. 
For an interval $J$ included in $[1, m]$ (resp. $[2, m]$) and $\ell \in J$, 
we define the operator $Z_{\ell}^{J}(\vec{z}) \, (\vec{z}\in\mathbb{C}^{m})$ 
acting on $U_{0}\otimes (U^{\otimes m})$ (resp. $U^{\otimes m}$) by 
\begin{align*}
Z_{\ell}^{J}(\vec{z})=\prod_{\begin{subarray}{c} i \in J \\ i<\ell \end{subarray}}f(z_{\ell}, z_{i})
\prod_{\begin{subarray}{c} i \in J \\ i>\ell \end{subarray}}f(z_{i}, z_{\ell})
\prod_{\begin{subarray}{c} i \in J \\ i<\ell \end{subarray}}^{\curvearrowleft}Y_{i-1}(z_{\ell}, z_{i}).  
\end{align*}

Hereafter we set 
\begin{align*}
u(\nu_{1}, \ldots , \nu_{m})=u_{\nu_{1}} \otimes \cdots \otimes u_{\nu_{m}} 
\end{align*}
for simplicity. 

\begin{lem}\label{lem:Yba}
Suppose that $r \ge b>a \ge 1$ and $m \ge 1$. 
Then it holds that 
\begin{align*}
\prod_{1\le i \le m}^{\curvearrowleft}(f(w, z_{i})Y_{i-1}(w, z_{i}))u(b, a^{m})=q^{m}u(a^{m}, b)+
\sum_{\ell=1}^{m}g(w, z_{\ell})Z_{\ell}^{[1, m]}(\vec{z}) u(b, a^{m})
\end{align*} 
on $U_{0}\otimes (U^{\otimes m})$. 
\end{lem}

\begin{proof}
The proof is by induction on $m$. 
When $m=1$ it follows from the definition of $Y_{0}(w, z)$. 
Suppose that it holds for a positive integer $m$. We have 
\begin{align}
& \label{eq:Yba}
\prod_{1\le i \le m+1}^{\curvearrowleft}(f(w, z_{i})Y_{i-1}(w, z_{i}))u(b, a^{m+1}) \\ 
&=\prod_{2\le i \le m+1}^{\curvearrowleft}(f(w, z_{i})Y_{i-1}(w, z_{i}))
\left\{q\,u(a, b, a^{m})+g(w, z_{1}) u(b, a^{m+1}) \right\} 
\nonumber \\ 
&=q \, \prod_{2\le i \le m+1}^{\curvearrowleft}(f(w, z_{i})Y_{i-1}(w, z_{i}))u(a, b, a^{m})+
g(w, z_{1})\prod_{i=2}^{m+1}f(z_{i}, w)  u(b, a^{m+1}). 
\nonumber 
\end{align}
By the induction hypothesis the first term in the right hand side is equal to 
\begin{align}
q\left\{q^{m}\, u(a^{m}, b)+
\sum_{\ell=2}^{m+1}g(w, z_{\ell})Z_{\ell}^{[2, m+1]}(\vec{z}) u(a, b, a^{m})  
\right\}.
\label{eq:Yba-2}
\end{align}
Now use the equality 
\begin{align*}
u(a, b, a^{m})=q^{-1}\left\{ 
f(z_{\ell}, z_{1})Y_{0}(z_{\ell}, z_{1})-g(z_{\ell}, z_{1})\right\}u(b, a^{m+1})    
\end{align*}
to substitute for the second term of \eqref{eq:Yba-2},  
and we see that the right hand side of \eqref{eq:Yba} is equal to  
\begin{align*}
& 
q^{m+1}\, u(a^{m}, b)+
\sum_{\ell=2}^{m+1}g(w, z_{\ell})Z_{\ell}^{[1, m+1]}(\vec{z}) u(b, a^{m+1})  \\ 
&+\left\{ 
g(w, z_{1})\prod_{i=2}^{m+1}f(z_{i}, w)-\sum_{\ell=2}^{m+1}
g(w, z_{\ell})g(z_{\ell}, z_{1})\prod_{\begin{subarray}{c} i=2 \\ i\not=\ell \end{subarray}}^{m+1}
f(z_{i}, z_{\ell})  
\right\}u(b, a^{m+1}). 
\end{align*}
The coefficient of the third term above is equal to 
$g(w, z_{1})\prod_{i=2}^{m+1}f(z_{i}, z_{1})$. 
Thus we obtain the desired equality for $m+1$. 
\end{proof}

\begin{lem}\label{lem:back-b}
Suppose that $r \ge b>a \ge 1$ and 
$k_{a}, k_{a-1}, \ldots , k_{1}$ are non-negative integers. 
Define the intervals $J_{p} \, (1\le p\le a)$ by 
\begin{align}
J_{p}=\{ i \in \mathbb{Z} \, | \, 
k_{a}+\cdots +k_{p+1}<i\le k_{a}+\cdots +k_{p} \}.  
\label{eq:def-J_{p}-from-a}
\end{align}
Then it holds that 
\begin{align*}
&
q^{\sum_{p=2}^{a}k_{p}}\sum_{\ell(1)\in J_{1}}g(w, z_{\ell(1)})Z_{\ell(1)}^{J_{1}}(\vec{z}) 
u(a^{k_{a}}, \ldots , 2^{k_{2}}, b, 1^{k_{1}}) \\ 
&=\sum_{t=1}^{a}
\sum_{a\ge p(t)>\cdots >p(2)>p(1)=1}
\sum_{\ell(t) \in J_{p(t)}} \cdots \sum_{\ell(1) \in J_{p(1)}} 
(-1)^{t-1}\prod_{s=1}^{t}g(z_{\ell(s-1)}, z_{\ell(s)}) \\ 
& \quad {}\times 
\prod_{1\le s \le t}^{\curvearrowleft}
Z_{\ell(s)}^{J_{p(s+1)-1}\cup \cdots \cup J_{p(s)}}(\vec{z})
u(b, a^{k_{a}}, \ldots , 1^{k_{1}})
\end{align*}
on $U_{0}\otimes (U^{\otimes (k_{a}+\cdots +k_{1})})$, 
where we set $z_{\ell(0)}=w$ and $J_{p(t+1)-1}=J_{a}$ in the right hand side.  
\end{lem}

\begin{proof}
We proceed by induction on $a$. 
If $a=1$, it is trivial. 
Suppose that it holds for an integer $a$ such that $b-2\ge a \ge 1$. 
By the induction hypothesis we have 
\begin{align*}
& 
q^{\sum_{p=2}^{a+1}k_{p}}\sum_{\ell(1)\in J_{1}}g(w, z_{\ell(1)})Z_{\ell(1)}^{J_{1}}(\vec{z}) 
u((a+1)^{k_{a+1}}, a^{k_{a}}, \ldots , 2^{k_{2}}, b, 1^{k_{1}}) \\ 
&=q^{k_{a+1}}\sum_{t=1}^{a}
\sum_{a\ge p(t)>\cdots >p(2)>p(1)=1}
\sum_{\ell(t) \in J_{p(t)}} \cdots \sum_{\ell(1) \in J_{p(1)}} 
(-1)^{t-1}\prod_{s=1}^{t}g(z_{\ell(s-1)}, z_{\ell(s)}) \\ 
& \quad {}\times 
\prod_{1\le s \le t}^{\curvearrowleft}
Z_{\ell(s)}^{J_{p(s+1)-1}\cup \cdots \cup J_{p(s)}}(\vec{z})
u((a+1)^{k_{a+1}}, b, a^{k_{a}}, \ldots , 1^{k_{1}}). 
\end{align*} 
The vector 
$q^{k_{a+1}}u((a+1)^{k_{a+1}}, b, a^{k_{a}}, \ldots , 1^{k_{1}})$ in the right hand side is equal to 
\begin{align*}
& 
\left\{ 
\prod_{i \in J_{a+1}}^{\curvearrowleft}(f(z_{\ell(t)}, z_{i})Y_{i-1}(z_{\ell(t)}, z_{i}))-
\sum_{\ell(t+1) \in J_{a+1}}g(z_{\ell(t)}, z_{\ell(t+1)})Z_{\ell(t+1)}^{J_{a+1}}(\vec{z})
\right\} \\ 
& \quad {}\times 
u(b, (a+1)^{k_{a+1}}, a^{k_{a}}, \ldots , 1^{k_{1}})  
\end{align*}
because of Lemma \ref{lem:Yba}. 
Substituting it we see that the equality to prove is also true for $a+1$. 
\end{proof}

Combining Lemma \ref{lem:Yba} and Lemma \ref{lem:back-b}, 
we obtain the following formula. 

\begin{cor}\label{cor:back-b}
Suppose that $r>a \ge 1$ and 
$k_{a}, k_{a-1}, \ldots , k_{1}$ are non-negative integers. 
Define the intervals $J_{p} \, (1\le p\le a)$ by \eqref{eq:def-J_{p}-from-a}. 
Then it holds that 
\begin{align*}
& 
q^{\sum_{p=2}^{a}k_{p}}\prod_{i \in J_{1}}^{\curvearrowleft}
(f_{i}(w, z_{i})Y_{i-1}(w, z_{i})) 
u(a^{k_{a}}, \ldots , 2^{k_{2}}, a+1, 1^{k_{1}}) \\ 
&=q^{\sum_{p=1}^{a}k_{p}}u(a^{k_{a}}, \ldots , 1^{k_{1}}, a+1)\\ 
&+\sum_{t=1}^{a} 
\sum_{a\ge p(t)>\cdots >p(2)>p(1)=1}
\sum_{\ell(t) \in J_{p(t)}} \cdots \sum_{\ell(1) \in J_{p(1)}} 
(-1)^{t-1}\prod_{s=1}^{t}g(z_{\ell(s-1)}, z_{\ell(s)}) \\ 
& \quad {}\times 
\prod_{1\le s \le t}^{\curvearrowleft}
Z_{\ell(s)}^{J_{p(s+1)-1}\cup \cdots \cup J_{p(s)}}(\vec{z})
u(a+1, a^{k_{a}}, \ldots , 1^{k_{1}}). 
\end{align*} 
on $U_{0}\otimes (U^{\otimes (k_{a}+\cdots +k_{1})})$, 
where we set $z_{\ell(0)}=w$ and $J_{p(t+1)-1}=J_{a}$ in the right hand side.  
\end{cor}

\begin{prop}\label{prop:recursion-h-prepare} 
Suppose that $r>a \ge 1$ and 
$k_{a}, k_{a-1}, \ldots , k_{1}$ are non-negative integers. 
Define the intervals $J_{p} \, (1\le p \le a)$ by \eqref{eq:def-J_{p}-from-a}. 
Then it holds that 
\begin{align}
& \label{eq:recursion-h-prepare-0} 
\prod_{1 \le i \le k_{a}+\cdots +k_{1}}^{\curvearrowleft}\!\!\!\!(f(w, z_{i})Y_{i-1}(w, z)) \, 
u(a+1, a^{k_{a}}, \ldots , 1^{k_{1}}) \\ 
&=q^{\sum_{p=1}^{a}k_{p}} u(a^{k_{a}}, \ldots , 1^{k_{1}}, a+1) 
\nonumber \\ 
&+\sum_{t=1}^{a}\sum_{a\ge p(t)>\cdots >p(1)\ge 1}q^{\sum_{b=1}^{p(1)-1}k_{b}}
\sum_{\ell(t) \in J_{p(t)}} \cdots \sum_{\ell(1) \in J_{p(1)}} 
\prod_{s=1}^{t}g(z_{\ell(s+1)}, z_{\ell(s)}) 
\nonumber \\ 
& \quad {}\times 
\prod_{1\le s \le t}^{\curvearrowleft}
Z_{\ell(s)}^{J_{p(s+1)-1} \cup \cdots \cup J_{p(s)}}(\vec{z})
u(a+1, a^{k_{a}}, \ldots , p(1)^{k_{p(1)}-1}, \ldots , 1^{k_{1}}, p(1)). 
\nonumber
\end{align} 
on $U_{0}\otimes (U^{\otimes (k_{a}+\cdots +k_{1})})$, 
where we set $z_{\ell(0)}=w$ and $J_{p(t+1)-1}=J_{a}$ in the right hand side.  
\end{prop}

\begin{proof}
The proof is by induction on $a$. 
If $a=1$, it follows from Lemma \ref{lem:Yba}. 
Suppose that it holds for an integer $a$ such that $r-1>a\ge 1$.  
Using the induction hypothesis we see that 
the left hand side of \eqref{eq:recursion-h-prepare-0} with $a$ replaced by $a+1$ is equal to  
the image of the vector 
\begin{align*}
& 
q^{\sum_{p=2}^{a+1}k_{p}} 
\biggl[
u((a+1)^{k_{a+1}}, \ldots , 2^{k_{1}}, a+2, 1^{k_{1}}) \\ 
&+\sum_{t=1}^{a}\sum_{a+1\ge p(t)>\cdots >p(1)\ge 2}q^{\sum_{b=2}^{p(1)-1}k_{b}}
\sum_{\ell(t) \in J_{p(t)}} \cdots \sum_{\ell(1) \in J_{p(1)}} 
\prod_{s=1}^{t}g(z_{\ell(s+1)}, z_{\ell(s)})   
\nonumber \\ 
& \quad {}\times 
\prod_{1\le s \le t}^{\curvearrowleft}\!\!
Z_{\ell(s)}^{J_{p(s+1)-1} \cup \cdots \cup J_{p(s)}}(\vec{z}) \,  
u(a+2, (a+1)^{k_{a}}, \ldots , p(1)^{k_{p(1)}-1}, \ldots , 2^{k_{2}}, p(1), 1^{k_{1}}) 
\biggr], 
\nonumber 
\end{align*}
where $z_{\ell(t+1)}=w$ and $J_{p(t+1)-1}=J_{a+1}$, 
under the operator $\prod_{i \in J_{1}}^{\curvearrowleft}(f(w, z_{i})Y_{i-1}(w, z))$.  
Now use Corollary \ref{cor:back-b} to calculate the image, 
and we see that it is equal to the sum of the following four vectors: 
\begin{align*}
& 
X_{I}=q^{\sum_{p=1}^{a+1}k_{a}}u((a+1)^{k_{a+1}}, \ldots , 1^{k_{1}}, a+2),  \\
& 
X_{II}=\sum_{t=1}^{a+1}\sum_{a+1 \ge p(t)>\cdots >p(2)>p(1)=1}
\sum_{\ell(t)\in J_{p(t)}}\cdots \sum_{\ell(1)\in J_{p(1)}}\!\!\! 
(-1)^{t-1}g(w, z_{\ell(1)}) \\ 
&\qquad {}\times \prod_{s=1}^{t-1}g(z_{\ell(s)}, z_{\ell(s+1)}) 
\prod_{1\le s \le t}^{\curvearrowleft}
Z_{\ell(s)}^{J_{p(s+1)-1} \cup \cdots \cup J_{p(s)}}(\vec{z}) 
u(a+2, (a+1)^{k_{a+1}}, \ldots , 1^{k_{1}}), \\ 
& 
X_{III}=\sum_{t=1}^{a}\sum_{a+1 \ge p(t)>\cdots >p(1)\ge 2}\!\!\!
q^{\sum_{p=1}^{p(1)-1}k_{p}}\!\!\!
\sum_{\ell(t)\in J_{p(t)}}\cdots \sum_{\ell(1)\in J_{p(1)}}\!\!\!
g(w, z_{\ell(t)}) \prod_{s=1}^{t-1}g(z_{\ell(s+1)}, z_{\ell(s)}) \\ 
&\qquad {}\times 
\prod_{1\le s \le t}^{\curvearrowleft}
Z_{\ell(s)}^{J_{p(s+1)-1} \cup \cdots \cup J_{p(s)}}(\vec{z}) 
u(a+2, (a+1)^{k_{a+1}}, \ldots , p(1)^{k_{p(1)}-1}, \ldots , 1^{k_{1}}, p(1)), 
\end{align*}
where $J_{p(t+1)-1}=J_{a+1}$, and 
\begin{align*}
& 
X_{IV}=\sum_{t=2}^{a+1}\sum_{a+1 \ge p(t)>\cdots >p(2)>p(1)=1}
\sum_{\ell(t)\in J_{p(t)}} \cdots \sum_{\ell(1) \in J_{p(1)}} 
K_{t}(w; z_{\ell(1)}, \ldots , z_{\ell(t)}) \\ 
& \quad {}\times 
\prod_{1\le s \le t}^{\curvearrowleft}
Z_{\ell(s)}^{J_{p(s+1)-1}\cup \cdots \cup J_{p(s)}}(\vec{z})
u(a+2, (a+1)^{k_{a+1}}, \ldots , 1^{k_{1}}), 
\end{align*}
where 
\begin{align*}
& 
K_{t}(w; z_{\ell(1)}, \ldots , z_{\ell(t)}) \\ 
&=\sum_{s=0}^{t-2}
g(w, z_{\ell(1)})g(w, z_{\ell(t)})\prod_{i=1}^{s}g(z_{\ell(i)}, z_{\ell(i+1)})
\prod_{i=s+2}^{t-1}g(z_{\ell(i+1)}, z_{\ell(i)}). 
\end{align*}
{}From the definition of $g$, we find that 
\begin{align*}
K_{t}(w; z_{\ell(1)}, \ldots , z_{\ell(t)})=\frac{g(w, z_{\ell(1)})g(w, z_{\ell(t)})}{g(z_{\ell(t)}, z_{\ell(1)})}
\prod_{s=1}^{t-1}g(z_{\ell(s+1)}, z_{\ell(s)}) 
\end{align*}
for $t \ge 2$. 
Substitute it for $X_{IV}$ and add $X_{II}$ and $X_{IV}$ using
\begin{align*}
& 
(-1)^{t-1}\prod_{s=1}^{t-1}g(z_{\ell(s)}, z_{\ell(s+1)})+
\frac{g(w, z_{\ell(t)})}{g(z_{\ell(t)}, z_{\ell(1)})}\prod_{s=1}^{t-1}g(z_{\ell(s+1)}, z_{\ell(s)}) \\ 
&=\frac{g(w, z_{\ell(t)})}{g(w, z_{\ell(1)})}\prod_{s=1}^{t-1}g(z_{\ell(s+1)}, z_{\ell(s)}).  
\end{align*}
As a result we get the desired equality for $a+1$. 
\end{proof}

Now let us prove the recurrence relation \eqref{eq:recursion} for $h^{\vec{\mu}}_{\vec{z}}$. 

\begin{proof}[Proof of Theorem \ref{thm:recursion} for $h^{\vec{\mu}}_{\vec{z}}$]
Set $J_{p}=\{ i \in \mathbb{Z} \, | \, k_{r}+\cdots +k_{p+1}<i\le k_{r}+\cdots +k_{p}\}$ for 
$1 \le p \le r$.  
In the definition of the value $h^{r^{k_{r}}, \ldots , 1^{k_{1}}}_{\vec{z}}(x_{1}, \ldots , x_{k-1}, 1)$, 
decompose the sum over $\tau \in \mathfrak{S}_{k}$ into $k$ parts 
with $\tau^{-1}(k)=\ell \, (1 \le \ell \le k)$. 
Change the running index $\tau$ to $\tau'=\tau(\sigma_{\ell} \sigma_{\ell+1} \cdots \sigma_{k-1})$, 
where $\sigma_{i}=(i, i+1)$ is the transposition.   
Then $\tau'(k)=k$ and we can regard $\tau'$ as an element of $\mathfrak{S}_{k-1}$. 
Thus we see that 
\begin{align}
& \label{eq:recursion-h-1}
h^{r^{k_{r}}, \ldots , 1^{k_{1}}}_{\vec{z}}(x_{1}, \ldots , x_{k-1}, 1)=
\prod_{1\le i<j \le k}f(z_{i}, z_{j}) \\ 
& \quad {}\times 
\sum_{p=1}^{r}\sum_{\ell \in J_{p}}\sum_{\tau' \in \mathfrak{S}_{k-1}}
\frac{z_{\ell}}{1+z_{\ell}}\prod_{i=1}^{\ell-1}\left(\frac{z_{i}}{1+z_{i}}\right)^{x_{\tau'(i)}}
\prod_{i=\ell+1}^{k}\left(\frac{z_{i}}{1+z_{i}}\right)^{x_{\tau'(i-1)}} 
\nonumber \\ 
& \qquad {}\times 
\phi(\tau'; z_{1}, \ldots , z_{\ell-1}, z_{\ell+1}, \ldots , z_{k}) 
\phi(\sigma_{k-1}\cdots \sigma_{\ell}; \vec{z}) 
u(r^{k_{r}}, \ldots , 1^{k_{1}}). 
\nonumber 
\end{align} 
For $\ell \in J_{p}$ it holds that 
\begin{align*}
\phi(\sigma_{k-1}\cdots \sigma_{\ell}; \vec{z}) 
u(r^{k_{r}}, \ldots , 1^{k_{1}})=
\prod_{\begin{subarray}{c} i \in J_{p} \\ i>\ell \end{subarray}}\frac{f(z_{i}, z_{\ell})}{f(z_{\ell}, z_{i})}
\prod_{i \in J_{p+1} \cup \cdots \cup J_{1}}^{\curvearrowleft} \!\!\! Y_{i-1}(z_{\ell}, z_{i}) \, 
u(r^{k_{r}}, \ldots , 1^{k_{1}}).  
\end{align*}
{}From Proposition \ref{prop:recursion-h-prepare}, we have 
\begin{align*}
& 
\prod_{i \in J_{p+1} \cup \cdots \cup J_{1}}^{\curvearrowleft} \!\!\! Y_{i-1}(z_{\ell}, z_{i}) \, 
u(r^{k_{r}}, \ldots , 1^{k_{1}})=
q^{\sum_{b=1}^{p-1}k_{b}}u(r^{k_{r}}, \ldots , p^{k_{p}-1}, \ldots , 1^{k_{1}}, p) \\ 
&\qquad {}+\sum_{t=1}^{p-1}\sum_{p-1\ge p(t-1)>\cdots >p(0) \ge 1}q^{\sum_{b=1}^{p(0)-1}k_{b}}
\sum_{\ell(t-1) \in J_{p(t-1)}} \cdots \sum_{\ell(0) \in J_{p(0)}} 
\prod_{s=1}^{t}g(z_{\ell(s)}, z_{\ell(s-1)}) \\ 
&\qquad  \quad {}\times 
\prod_{0\le s \le t-1}^{\curvearrowleft}Z_{\ell(s)}^{J_{p(s+1)-1} \cup \cdots \cup J_{p(s)}}(\vec{z}) 
u(r^{k_{r}}, \ldots , p(0)^{k_{p(0)}-1}, \ldots , 1^{k_{1}}, p(0)),  
\end{align*}
where $J_{p(t)-1}=J_{p-1}$. 
Thus we see that the right hand side of \eqref{eq:recursion-h-1} is the sum of the following 
two terms: 
\begin{align*}
S_{I}=\sum_{a=1}^{r}q^{\sum_{b=1}^{a-1}k_{b}}\sum_{\ell \in J_{a}}\frac{z_{\ell}}{1+z_{\ell}} 
\prod_{\begin{subarray}{c} i \in J_{r} \cup \cdots \cup J_{a} \\ i\not=\ell \end{subarray}} \!\!\! 
f(z_{i}, z_{\ell}) \, 
h_{z_{1}, \ldots , z_{\ell-1}, z_{\ell+1}, \ldots , z_{k}}^{r^{k_{r}}, \ldots , a^{k_{a}-1}, \ldots , 1^{k_{1}}} 
(x_{1}, \ldots , x_{k-1}) \otimes u_{a}
\end{align*}
and 
\begin{align*}
S_{II}
&=\prod_{1\le i<j \le k}f(z_{i}, z_{j}) \\ 
&\times 
\sum_{a=1}^{r-1}q^{\sum_{b=1}^{a-1}k_{b}}
\sum_{t=1}^{r-a}\sum_{r \ge p(t)>\cdots >p(1)>p(0)=a}
\sum_{\ell(t) \in J_{p(t)}} \cdots \sum_{\ell(0) \in J_{p(0)}} 
\prod_{\begin{subarray}{c} i \in J_{p(t)} \\ i>\ell(t) \end{subarray}}
\frac{f(z_{i}, z_{\ell(t)})}{f(z_{\ell(t)}, z_{i})} \\ 
&\times 
\sum_{\tau' \in \mathfrak{S}_{k-1}}\frac{z_{\ell(t)}}{1+z_{\ell(t)}} 
\prod_{i=1}^{\ell(t)-1}\left(\frac{z_{i}}{1+z_{i}}\right)^{x_{\tau'(i)}}
\prod_{i=\ell(t)+1}^{k}\left(\frac{z_{i}}{1+z_{i}}\right)^{x_{\tau'(i-1)}}
\prod_{s=1}^{t}g(z_{\ell(s)}, z_{\ell(s-1)}) \\ 
& \quad {}\times 
\phi(\tau'; z_{1}, \ldots , z_{\ell(t)-1}, z_{\ell(t)+1}, \ldots , z_{k})
\prod_{0\le s \le t-1}^{\curvearrowleft}Z_{\ell(s)}^{J_{p(s+1)-1} \cup \cdots \cup J_{p(s)}}(\vec{z}) \\ 
& \quad {}\times 
u(r^{k_{r}}, \ldots , a^{k_{a}-1}, \ldots , 1^{k_{1}}, a).   
\end{align*}
We rewrite $S_{II}$ as follows. 
Lemma \ref{lem:phi-decompose} implies that 
\begin{align*}
& 
\prod_{1\le i<j \le k}\!\!\!\!\!f(z_{i}, z_{j})
\prod_{\begin{subarray}{c} i \in J_{p(t)} \\ i>\ell(t) \end{subarray}}\!\!
\frac{f(z_{i}, z_{\ell(t)})}{f(z_{\ell(t)}, z_{i})}\,
\phi(\tau'; z_{1}, \ldots , z_{\ell(t)-1}, z_{\ell(t)+1}, \ldots , z_{k})\!\!
\prod_{0\le s \le t-1}^{\curvearrowleft}\!\!\!
Z_{\ell(s)}^{J_{p(s+1)-1} \cup \cdots \cup J_{p(s)}}(\vec{z}) \\ 
&=\prod_{1\le i<j\le k-1}f(w_{i}, w_{j})
\prod_{s=0}^{t}(
\prod_{\begin{subarray}{c} i\in J_{p(s+1)-1}\cup \cdots \cup J_{p(s)} \\ i\not=\ell(s) \end{subarray}}
f(z_{i}, z_{\ell(s)})
) \, 
\phi(\tilde{\tau}; \vec{w}), 
\end{align*}
where $\tilde{\tau} \in \mathfrak{S}_{k-1}$ and $\vec{w}=(w_{1}, \ldots , w_{k-1})$ are defined by 
\begin{align*}
\tilde{\tau}=\tau' \prod_{0\le s <t}^{\curvearrowleft}
(\prod_{\begin{subarray}{c} i \in J_{p(s)} \\ i<\ell(s) \end{subarray}}^{\curvearrowright} \sigma_{i-1}) 
\end{align*}
and 
\begin{align}\label{eq:recursion-h-2}
w_{i}=\left\{ 
\begin{array}{ll}
z_{i} & (i\not\in J_{p(t)}\cup \cdots \cup J_{p(0)} \,\, \hbox{and} \,\, i<k_{r}+\cdots+k_{p(0)}) \\ 
z_{i} & (i \in J_{p(s)} \,\, \hbox{and} \,\, i<\ell(s), \, 0\le s \le t) \\ 
z_{i+1} & (i \in J_{p(s)} \,\, \hbox{and} \,\, i\ge \ell(s), \, 1\le s \le t) \\ 
z_{\ell(s-1)} & (i=k_{r}+\cdots +k_{p(s)}, \, 1\le s \le t) \\ 
z_{i+1} & (k_{r}+\cdots+k_{p(0)}\le  i \le k-1). 
\end{array}
\right. 
\end{align}
Then it holds that 
\begin{align*}
\prod_{i=1}^{\ell(t)-1}\left(\frac{z_{i}}{1+z_{i}}\right)^{x_{\tau'(i)}}
\prod_{i=\ell(t)+1}^{k}\left(\frac{z_{i}}{1+z_{i}}\right)^{x_{\tau'(i-1)}}=
\prod_{i=1}^{k-1}\left(\frac{w_{i}}{1+w_{i}}\right)^{x_{\tilde{\tau}(i)}}.  
\end{align*}
Now change the running index $\tau'$ to $\tilde{\tau}$ in $S_{II}$ and we find that 
\begin{align*}
\mathrm{II}&=\sum_{a=1}^{r-1}q^{\sum_{b=1}^{a-1}k_{b}}
\sum_{t=1}^{r-a}\sum_{r \ge p(t)>\cdots >p(1)>p(0)=a}
\sum_{\ell(t) \in J_{p(t)}} \cdots \sum_{\ell(0) \in J_{p(0)}} 
\frac{z_{\ell(t)}}{1+z_{\ell(t)}} \\ 
&\times \prod_{s=1}^{t}g(z_{\ell(s)}, z_{\ell(s-1)})
\prod_{s=0}^{t}(
\prod_{\begin{subarray}{c} i\in J_{p(s+1)-1}\cup \cdots \cup J_{p(s)} \\ i\not=\ell(s) \end{subarray}}
f(z_{i}, z_{\ell(s)})
) \\ 
&\times 
h^{r^{k_{r}}, \ldots , a^{k_{a}-1}, \ldots , 1^{k_{1}}}_{\vec{w}}(x_{1}, \ldots , x_{k-1}) \otimes u_{a},  
\end{align*}
where $\vec{w}$ is defined by \eqref{eq:recursion-h-2}. 
Proposition \ref{prop:symmetry} (1) implies that 
\begin{align*}
h^{r^{k_{r}}, \ldots , a^{k_{a}-1}, \ldots , 1^{k_{1}}}_{\vec{w}}(x_{1}, \ldots , x_{k-1})=
h^{r^{k_{r}}, \ldots , a^{k_{a}-1}, \ldots , 1^{k_{1}}}_{\vec{z}(\ell(t), \ldots , \ell(0))}
(x_{1}, \ldots , x_{k-1}).   
\end{align*}
Now the term $S_{I}$ compensates for the lack of the terms in $S_{II}$ with $t=0$, 
and we get the desired recurrence relation for $h^{\vec{\mu}}_{\vec{z}}$. 
\end{proof}

\subsection{Proof for $\psi_{\vec{z}}^{\vec{\mu}}$}

Next we show that the function $\psi_{\vec{z}}^{\vec{\mu}}$ also satisfies the recurrence relation 
\eqref{eq:recursion}.  
Fix $\vec{x}=(x_{1}, \ldots , x_{k}) \in L_{k}^{+}$ and an integer $M$ 
such that $x_{k}=1$ and $M \ge x_{1}$. 
Then 
\begin{align*}
& \quad 
\psi_{\vec{z}}^{r^{k_{r}}, \ldots , 1^{k_{1}}}(\vec{x})=(1-q^{2})^{-k} \\ 
&\times 
\sum_{\vec{\nu} \in I_{k_{1}, \ldots , k_{r}}} 
\prod_{i=1}^{k}\left(\frac{1}{1+z_{i}}\right)^{M} 
\langle 
\prod_{1\le i \le k}^{\curvearrowright}C_{\mu_{i}}^{[1, M]}(z_{i})
\prod_{1\le i \le k}\beta_{\nu_{i}, x_{i}}^{*}
\rangle_{[1, M]}
q^{t(\vec{\nu})-t(\vec{\mu})} u_{\nu_{1}} \otimes \cdots \otimes u_{\nu_{k}}.     
\end{align*}
Hereafter we often omit the upper index $[1, M]$ of $A^{[1, M]}(z)$ and $C_{a}^{[1, M]}(z) \, (1\le a \le r)$. 

\begin{lem}\label{lem:commrel-C-beta*}
It holds that 
$C_{b}(z)\beta_{a, 1}^{*}=\beta_{a, 1}^{*} C_{b}(z)$ if $a<b$, and 
$C_{b}(z)\beta_{a, 1}^{*}=q^{2}\beta_{a, 1}^{*} C_{b}(z)$ if $a>b$. 
\end{lem}

\begin{proof}
It follows from Proposition \ref{prop:recursion-AC}. 
\end{proof}

For $1\le a \le r$ we set 
\begin{align*}
\tilde{A}_{a}(z)=q^{2\sum_{p=a}^{r}N_{p, 1}}A^{[2, M]}(z).  
\end{align*}

\begin{lem}\label{lem:commrel-CCAtilde}
For $(z_{1}, \ldots , z_{m}) \in \mathbb{C}^{m}, w \in \mathbb{C}$ and $1\le a \le b \le r$, it holds that 
\begin{align*}
\prod_{i=1}^{m}C_{b}(z_{i}) \cdot \tilde{A}_{a}(w)&\equiv 
\prod_{i=1}^{m}\!f(z_{i}, w) \, \tilde{A}_{a}(w) \prod_{i=1}^{m}C_{b}(z_{i}) \\ 
&+\sum_{\ell=1}^{m}\frac{z_{\ell}}{w}g(z_{\ell}, w)
\prod_{\begin{subarray}{c} i=1 \\ i\not=\ell \end{subarray}}^{m}f(z_{i}, z_{\ell}) \, 
\tilde{A}_{a}(z_{\ell}) C_{b}(w)
\prod_{\begin{subarray}{c} i=1 \\ i\not=\ell \end{subarray}}^{m}C_{a}(z_{i})  
\end{align*} 
modulo the right ideal $\sum_{p=1}^{b-1}\beta_{p, 1}^{*}\mathcal{B}^{[1, M]}$. 
\end{lem}

\begin{proof}
The proof is by induction on $m$. 
For $m=1$ it follows from \eqref{eq:recursion-C1}, \eqref{eq:commrel-AC} on the interval $[2, M]$ 
and the equality $f(z, w)+g(z, w)=q^{2}$.  
Suppose that it holds for a positive integer $m$. 
Now decompose $\prod_{i=1}^{m+1}C_{b}(z_{i}) \cdot \tilde{A}_{a}(w)= 
C_{b}(z_{1}) \prod_{i=2}^{m+1}C_{b}(z_{i}) \cdot \tilde{A}_{a}(w)$
and calculate it using the induction hypothesis. 
Note that $\beta_{p, 1}^{*} \, (1\le p<b)$ commutes with $C_{b}(z)$ 
because of Lemma \ref{lem:commrel-C-beta*}. 
Then we obtain the desired formula for $m+1$ by using the equality 
\begin{align*}
g(z_{1}, w)\prod_{i=2}^{m+1}f(z_{i}, w)+
\sum_{\ell=2}^{m+1}g(z_{\ell}, w)g(z_{1}, z_{\ell})
\prod_{\begin{subarray}{c} i=2 \\ i\not=\ell \end{subarray}}^{m+1}f(z_{i}, z_{\ell})=
g(z_{1}, w)\prod_{i=2}^{m+1}f(z_{i}, z_{1}).   
\end{align*}
\end{proof}

\begin{prop}\label{prop:beta-star-move}
For $(z_{1}, \ldots , z_{m}) \in \mathbb{C}^{m}$ and $1\le a \le r$, it holds that 
\begin{align*}
\prod_{i=1}^{m}C_{a}(z_{i}) \cdot \beta_{a, 1}^{*}\equiv
\beta_{a, 1}^{*} \prod_{i=1}^{m}C_{a}(z_{i})+(1-q^2)\sum_{\ell=1}^{r}z_{\ell}
\prod_{\begin{subarray}{c} i=1 \\ i\not=\ell \end{subarray}}^{m}f(z_{i}, z_{\ell}) \cdot  
\tilde{A}_{a}(z_{\ell}) \prod_{\begin{subarray}{c} i=1 \\ i\not=\ell \end{subarray}}^{m}C_{a}(z_{i}) 
\end{align*} 
modulo the right ideal $\sum_{p=1}^{a-1}\beta_{p, 1}^{*}\mathcal{B}^{[1, M]}$. 
\end{prop}

\begin{proof}
The recurrence relation \eqref{eq:recursion-C1} implies that 
\begin{align*}
C_{a}(z)\beta_{a, 1}^{*}\equiv\beta_{a, 1}^{*}C_{a}(z)+(1-q^{2})z\tilde{A}_{a}(z). 
\end{align*}
Using the above relation repeatedly we find that 
\begin{align*}
\prod_{i=1}^{m}C_{a}(z_{i}) \cdot \beta_{a, 1}^{*}\equiv
\beta_{a, 1}^{*}\prod_{i=1}^{m}C_{a}(z_{i})+
(1-q^2)\sum_{\ell=1}^{m}z_{\ell}\prod_{i=1}^{\ell-1}C_{a}(z_{i}) \cdot \tilde{A}_{a}(z_{\ell}) 
\prod_{i=\ell+1}^{m}C_{a}(z_{i}).  
\end{align*}
Move $\tilde{A}_{a}(z_{\ell})$ in the right hand side to the left using Lemma \ref{lem:commrel-CCAtilde}. 
Then the equality
\begin{align*}
\prod_{i=1}^{\ell-1}f(z_{i}, z_{\ell})+\sum_{t=\ell+1}^{m}g(z_{\ell}, z_{t})
\prod_{\begin{subarray}{c} i=1 \\ i\not=\ell \end{subarray}}^{t-1}f(z_{i}, z_{\ell})=
\prod_{\begin{subarray}{c} i=1 \\ i\not=\ell \end{subarray}}^{m}f(z_{i}, z_{\ell}) 
\end{align*}
implies the desired formula. 
\end{proof}

Suppose that $1 \le a<r$ and $m_{r}, m_{r-1}, \ldots , m_{a+1}$ are non-negative integers. 
Set 
\begin{align*}
C^{(m_{r}, m_{r-1}, \ldots , m_{a+1})}(\vec{z})=
\prod_{a+1\le p \le r}^{\curvearrowleft}
\prod_{i=m_{r}+\cdots +m_{p+1}+1}^{m_{r}+\cdots +m_{p}}C_{p}(z_{i})  
\end{align*}
for $\vec{z} \in \mathbb{C}^{m_{r}+\cdots +m_{a+1}}$.  

\begin{lem}\label{lem:A-tilde-move}
Suppose that $1\le a<r$ and 
$k_{r}, k_{r-1}, \ldots , k_{a+1}$ are non-negative integers. 
For $a<p\le r$, set 
$J_{p}=\{ i \in \mathbb{Z} \,| \, k_{r}+\cdots +k_{p+1}<i \le k_{r}+\cdots +k_{p}\}$. 
Then it holds that 
\begin{align*}
& 
C^{(k_{r}, k_{r-1}, \ldots, k_{a+1})}(\vec{z})\tilde{A}_{a}(w) 
\equiv 
(\prod_{i \in J_{r}\cup \cdots \cup J_{a+1}}f(z_{i}, w) )
\tilde{A}_{a}(w) C^{(k_{r}, k_{r-1}, \ldots, k_{a+1})}(\vec{z}) \\ 
&+\sum_{t=1}^{r-a}
\sum_{r\ge p(t)>\cdots >p(1)>a} 
\sum_{\ell(t) \in J_{p(t)}}\cdots \sum_{\ell(1) \in J_{p(1)}}
\frac{z_{\ell(t)}}{w} \, \widetilde{A}_{a}(z_{\ell(t)}) 
\prod_{s=1}^{t}g(z_{\ell(s)}, z_{\ell(s-1)}) \\ 
&\quad {}\times 
\prod_{s=0}^{t}(
\prod_{\begin{subarray}{c} i \in J_{p(s+1)-1}\cup \cdots \cup J_{p(s)} \\ i\not=\ell(s) \end{subarray}}
f(z_{i}, z_{\ell(s)}) \, 
C^{(k_{r}, k_{r-1}, \ldots , k_{a+1})}(\vec{z}(\ell(t), \ldots , \ell(1); w)) 
\end{align*} 
modulo the right ideal $\sum_{p=1}^{r-1}\beta_{p, 1}^{*}\mathcal{B}^{[1, M]}$ for 
$\vec{z} \in \mathbb{C}^{k_{r}+\cdots +k_{a+1}}$, 
where $z_{\ell(0)}=w, J_{p(t+1)-1}=J_{r}, J_{p(0)}=J_{a+1}$ and 
$\vec{z}(\ell(t),  \ldots , \ell(1); w) \in \mathbb{C}^{k_{r}+\cdots +k_{a+1}}$ is defined by 
\begin{align*}
& 
\vec{z}(\ell(t),  \ldots , \ell(1); w)_{i}=\left\{
\begin{array}{ll}
z_{i} & (i\not=\ell(t), \ldots , \ell(1)) \\ 
z_{\ell(s-1)} & (i=\ell(s), 2\le s \le t) \\ 
w & (i=\ell(1)).  
\end{array}
\right.  
\end{align*}
\end{lem}

\begin{proof}
By reverse induction on $a$ using Lemma \ref{lem:commrel-CCAtilde}.  
\end{proof}

\begin{prop}\label{prop:beta-star-move-vac}
Suppose that $1\le a \le r$ and that  
$k_{r}, k_{r-1}, \ldots , k_{a+1} \in \mathbb{Z}_{\ge 0}$ and $k_{a} \in \mathbb{Z}_{\ge 1}$. 
Set 
$J_{p}=\{ i \in \mathbb{Z} \,| \, k_{r}+\cdots +k_{p+1}<i \le k_{r}+\cdots +k_{p}\} \, (a\le p \le r)$. 
Then it holds that 
\begin{align*}
& 
{}_{[1, M]}\langle \mathrm{vac} | 
C^{(k_{r}, \ldots , k_{a})}(\vec{z}) \beta_{a, 1}^{*} \\ 
&=(1-q^{2})\sum_{t=0}^{r-a}\sum_{r\ge p(t)>\cdots >p(1)>p(0)=a} 
\sum_{\ell(t) \in J_{p(t)}} \cdots \sum_{\ell(0) \in J_{p(0)}} 
z_{\ell(t)}(1+z_{\ell(t)})^{M-1} \\ 
& \qquad \qquad {}\times 
\prod_{s=1}^{t}g(z_{\ell(s)}, z_{\ell(s-1)}) 
\prod_{s=0}^{t}(
\prod_{\begin{subarray}{c} i \in J_{p(s+1)-1} \cup \cdots \cup J_{p(s)} \\ i\not=\ell(s) \end{subarray}}
f(z_{i}, z_{\ell(s)})) \\ 
& \qquad \qquad {}\times 
{}_{[1, M]}\langle \mathrm{vac} | C^{(k_{r}, \ldots , k_{a+1}, k_{a}-1)}( 
\vec{z}(\ell(t), \ldots , \ell(0))
)  
\end{align*} 
for $\vec{z} \in \mathbb{C}^{k_{r}+\cdots +k_{a}}$, where $J_{p(t+1)-1}=J_{r}$ and 
$\vec{z}(\ell(t), \ldots , \ell(0)) \in \mathbb{C}^{k_{r}+\cdots +k_{a}-1}$ is defined by \eqref{eq:def-z-ell}. 
\end{prop}

\begin{proof}
Note that ${}_{[1, M]}\langle \mathrm{vac} |\beta_{p, 1}^{*}=0$ for any $1\le p \le r$. 
{}From Lemma \ref{lem:commrel-C-beta*} and Proposition \ref{prop:beta-star-move} 
we see that the left hand side is equal to 
\begin{align*}
(1-q^2)\sum_{\ell \in J_{a}}z_{\ell}
\prod_{i \in J_{a}\setminus\{\ell\}}f(z_{i}, z_{\ell}) \cdot 
{}_{[1, M]}\langle \mathrm{vac} | 
C^{(k_{r}, \ldots , k_{a+1})}(\vec{z}) 
\tilde{A}_{a}(z_{\ell}) 
\prod_{i \in J_{a}\setminus\{\ell\}}C_{a}(z_{i}).  
\end{align*}
Move $\tilde{A}_{a}(z_{\ell})$ to the left by using Lemma \ref{lem:A-tilde-move}. 
{}From \eqref{eq:recursion-A} we have 
\begin{align*}
{}_{[1, M]}\langle \mathrm{vac} | \tilde{A}(w)=(1+w)^{M-1}{}_{[1, M]}\langle \mathrm{vac} |.  
\end{align*}
Thus we obtain the desired formula. 
\end{proof}

Now we prove the recurrence relation \eqref{eq:recursion} for $\psi^{\vec{\mu}}_{\vec{z}}$. 

\begin{proof}[Proof of Theorem \ref{thm:recursion} for $\psi^{\vec{\mu}}_{\vec{z}}$]
Set $\vec{\mu}=(r^{k_{r}}, \ldots , 1^{k_{1}})$. 
Decompose the sum over $\vec{\nu}$ in the definition of 
$\psi^{\vec{\mu}}_{\vec{z}}(x_{1}, \ldots , x_{k-1}, 1)$ into $r$ parts with $\nu_{k}=a \, (1\le a \le r)$. 
Then Lemma \ref{lem:commrel-C-beta*} implies that 
\begin{align*}
& 
\psi_{\vec{z}}^{\vec{\mu}}(x_{1}, \ldots , x_{k-1}, 1)=(1-q^{2})^{-k}\prod_{i=1}^{k}\frac{1}{(1+z_{i})^{M}} \\ 
&\times 
\sum_{a=1}^{r}q^{2\sum_{p=1}^{a-1}k_{p}}
\sum_{\begin{subarray}{c} \vec{\nu} \in I_{k_{r}, \ldots , k_{1}} \\ \nu_{k}=a \end{subarray}} 
\langle C^{(k_{r}, \ldots , k_{a})}(\vec{z})\beta_{a, 1}^{*} 
C^{(k_{a-1}, \ldots , k_{1})}(\overrightarrow{z[a]}) 
\prod_{i=1}^{k-1}\beta_{\nu_{i}, x_{i}}^{*}\rangle_{[1, M]} \\ 
& \quad {}\times 
q^{t(\vec{\mu})-t(\vec{\nu})}u_{\nu_{1}} \otimes \cdots \otimes u_{\nu_{k-1}} \otimes u_{a},  
\end{align*}
where $\overrightarrow{z[a]} \in \mathbb{C}^{k_{a-1}+\cdots +k_{1}}$ is defined by 
$\overrightarrow{z[a]}_{i}=z_{i+k_{r}+\cdots +k_{a}}$. 
Now move $\beta_{a, 1}^{*}$ to the left using Proposition \ref{prop:beta-star-move-vac}.  
Note that 
\begin{align*}
t(\vec{\mu})-t(\vec{\nu})=t(r^{k_{r}}, \ldots , a^{k_{1}-1}, \ldots , 1^{k_{1}})-t(\nu_{1}, \ldots , \nu_{k-1})-
\sum_{b=1}^{a-1}k_{b} 
\end{align*} 
for $\vec{\nu} \in I_{k_{r}, \ldots , k_{1}}$ such that $\nu_{k}=a$.  
Thus we get the relation \eqref{eq:recursion} for $\psi_{\vec{z}}^{\vec{\mu}}$. 
\end{proof}


\section*{Acknowledgments}

The research of the author is supported by 
JSPS KAKENHI Grant Number 26400106. 
The author thanks the referee for valuable comments.

\end{document}